\documentclass[pra,superscriptaddress,twocolumn]{revtex4-1}%
\usepackage{amsfonts}
\usepackage{amsmath}
\usepackage{amssymb}
\usepackage{physics}
\usepackage{graphicx}
\usepackage{footnote}
\usepackage[colorlinks=true,linkcolor=blue,citecolor=blue,plainpages=false,pdfpagelabels]%
{hyperref}
\usepackage[caption=false]{subfig}
\usepackage{overpic}%
\providecommand{\U}[1]{\protect\rule{.1in}{.1in}}

\newcommand{\ahat}{\overline{A}}

\providecommand{\U}[1]{\protect\rule{.1in}{.1in}}
\newtheorem{theorem}{Theorem}

\newtheorem{definition}[theorem]{Definition}

\newtheorem{notation}[theorem]{Notation}

\newtheorem{proposition}[theorem]{Proposition}

\newenvironment{proof}[1][Proof]{\noindent\textbf{#1.} }{\ \rule{0.5em}{0.5em}}
\begin{document}
\preprint{ }
\title{Conditional Mutual Information and Quantum Steering}
\author{Eneet Kaur}
\email{ekaur1@lsu.edu}
\affiliation{Hearne Institute for Theoretical Physics, Department of Physics and Astronomy,
Baton Rouge, Louisiana 70803, USA}
\author{Xiaoting Wang}
\email{xiaoting@lsu.edu}
\affiliation{Hearne Institute for Theoretical Physics, Department of Physics and Astronomy,
Baton Rouge, Louisiana 70803, USA}
\author{Mark M. Wilde}
\email{mwilde@lsu.edu }
\affiliation{Hearne Institute for Theoretical Physics, Department of Physics and Astronomy,
Baton Rouge, Louisiana 70803, USA}
\affiliation{Center for Computation and Technology, Louisiana State University, Baton
Rouge, Louisiana 70803, USA}
\keywords{}
\pacs{PACS number}

\begin{abstract}
Quantum steering has recently been formalized in the framework of a resource
theory of steering, and several quantifiers have already been introduced.
Here, we propose an information-theoretic quantifier for steering called
\textit{intrinsic steerability}, which uses conditional mutual information to
measure the deviation of a given assemblage from one having a local
hidden-state model. We thus relate conditional mutual information to quantum
steering and introduce monotones that satisfy certain desirable properties.
The idea behind the quantifier is to suppress the correlations that can be
explained by an inaccessible quantum system and then quantify the remaining
intrinsic correlations. A variant of the intrinsic steerability finds
operational meaning as the classical communication cost of sending the
measurement choice and outcome to an eavesdropper who possesses a purifying
system of the underlying bipartite quantum state that is being measured.

\end{abstract}
\volumeyear{ }
\volumenumber{ }
\issuenumber{ }
\eid{ }
\date{\today}
\startpage{1}
\endpage{102}
\maketitle

\section{Introduction}Quantum steering was first introduced by
Schr\"{o}dinger in 1935 \cite{Schroedinger1935} in order to formalize an
argument made by Einstein, Podolsky, and Rosen in \cite{Einstein1935}. It
refers to the following scenario: two parties called Alice and Bob share a bipartite quantum
state. Alice measures her system, which can have the effect of steering the
reduced state on Bob's system, depending on the measurement that she performs.
She thus can influence Bob's subsystem without having access to it. However,
Bob does not have any knowledge about the influence, nor can he detect it
unless Alice communicates the measurement that she performed and the outcome
of the measurement. For example, consider a maximally entangled singlet shared
by Alice and Bob. Alice can measure her system in either the Pauli $\sigma
_{Z}$ basis or the Pauli $\sigma_{X}$ basis. If she measures in the Pauli
$\sigma_{Z}$ basis, the resulting state of Bob's subsystem is represented as
the ensemble $\left\{ (\frac{1}{2},\op{1}{1}),(\frac{1}{2}%
,\op{0}{0})\right\}$. Alternatively, if she measures in the Pauli
$\sigma_{X}$ basis, the state of Bob's subsystem is represented as the
ensemble $\left\{ (\frac{1}{2},\op{+}{+}),(\frac{1}{2},\op{-}{-})\right\}$.

The notion of steering was formalized in \cite{Wiseman2006}, which defines it
in the context of an entanglement certification task, with Alice having access
to an untrusted device and Bob to a trusted quantum system. Alice's 
device can be thought of as a black box, which accepts a classical input $X$
and outputs a classical system $\overline{A}$. 
The mathematical description of the
relation between Alice's classical input $X$, her output $\overline{A}$, and Bob's
quantum system is called an \textit{assemblage}, whose formal definition we
recall later. 

The fact that Alice's system is classical and Bob's system is quantum in the
scenario of steering makes it natural to study in the context of one-sided
device-independent tasks such as quantum key distribution \cite{Branciard2012}
and randomness certification \cite{Law2014,Passaro2015}. Apart from this,
Ref.~\cite{Piani2014} demonstrated the usefulness of steering in a task called
sub-channel discrimination, which deals with determining the direction of the
evolution of a system. Consider a state $\rho$ that evolves according to a
channel $\mathcal{N}=\sum_{z}p_{Z}(z)\mathcal{N}_{z}$, which is equal to a
random selection of a channel $\mathcal{N}_{z}$ according to the probability
distribution $p_{Z}$. Then the information regarding which path the system
takes is known as sub-channel discrimination.

A framework for a resource theory of steering was introduced in
\cite{Gallego2015}, in which one-way classical communication from Bob to Alice
and local operations (1W-LOCC) are taken as free operations. In
this framework, Bob is also allowed to measure his system and communicate the
classical measurement outcome prior to the measurement choice by Alice
\cite[Definition~1]{Gallego2015}. Thus, he can influence the input to her
black box. See Figure \ref{1W-LOCCf} for a schematic representation. In the resource theory of steering, any steering monotone should
be non-increasing under 1W-LOCC and equal to zero if a given assemblage is
unsteerable. It is also desirable for the quantity to be convex. Several
steering quantifiers, including robustness of steering \cite{Skrzypczyk2013},
steerable weight \cite{Piani2014}, and relative entropy of steering
\cite{Gallego2015,Kaur2016}, have been defined and proven to be a
steering monotone.

One contribution of our paper is to introduce \textit{intrinsic
steerability} as a measure of steering. Intrinsic steerability uses
conditional mutual information to measure the deviation of a given assemblage
from one having a local hidden-state model. The idea behind the quantifier is
to suppress the correlations that can be explained by an inaccessible quantum
system and then quantify the remaining intrinsic correlations. We prove that
intrinsic steerability is monotone with respect to 1W-LOCC and also that it is
convex and superadditive in general.

We also consider a simpler, restricted class of free operations in which Bob
cannot influence Alice's input to her black box. In considering this
restricted class, we are motivated by practical, relativistic constraints that
can potentially limit the performance of Alice and Bob's quantum devices in
any quantum steering protocol. Typically, in any such protocol, Alice, Bob,
and the source of their systems are spatially separated, and furthermore,
their quantum devices typically have a finite coherence time. If Alice were to
wait to receive a signal from Bob before taking any action on her system, the
performance of her device could potentially get much worse than it would be if
she were simply instead to input to her system as soon as she receives it from
the source. This perspective motivates a restricted class of 1W-LOCC
operations in which any classical communication from Bob reaches Alice only
after she has received the output $\overline{A}$ from her black box. We refer to
these free operations as \textit{restricted 1W-LOCC}.

We define the \textit{restricted intrinsic steerability} as a steering
quantifier, which is relevant for the aforementioned restricted class of
1W-LOCC operations. We prove that, along with it being a monotone with respect
to restricted 1W-LOCC and satisfying the properties mentioned above, it also
satisfies additivity and monogamy. To our knowledge, this is the first measure shown to be monogamous and additive with respect to tensor products of assemblages. 

Our approach to defining these steering quantifiers is inspired by the
approach of \cite{Tucci2002} to quantifying non-Markovianity in Bayesian
networks, which in turn bears connections to the squashed-entanglement measure
\cite{Christandl2003} and the intrinsic-information quantifier from classical
information theory \cite{MW99}. To see this, consider that correlations in any
unsteerable assemblage can be explained by a hidden variable, which implies
that such an assemblage has a Markov-chain structure. Assemblages with this
structure thus have zero conditional mutual information when conditioning on
the shared variable \cite{Hayden2003}, where we recall that the conditional
mutual information of a tripartite quantum state $\sigma_{KLM}$ is defined as
\begin{multline}
I(K;L|M)_{\sigma}:=H(KM)_{\sigma}+H(LM)_{\sigma}\\-H(KLM)_{\sigma} 
-H(M)_{\sigma}
\end{multline}
and $H(G)_{\omega}:=-\operatorname{Tr}(\omega_{G}\log
_{2}\omega_{G})$ denotes the quantum entropy of the state $\omega_{G}$ defined
on system $G$ (note that throughout this paper, we use the binary logarithm in
the definition of entropy). So our primary idea is to take a non-signaling
extension of an assemblage, remove the correlations which can be explained by
a shared variable (by conditioning), and then quantify the remaining intrinsic
correlations. 

\section{Preliminaries}

We begin by reviewing the framework of quantum steering as discussed in \cite{Gallego2015}. Let $\rho_{AB}$ be a bipartite quantum state shared
by Alice and Bob. Suppose that Alice performs a measurement labeled by
$x\in\mathcal{X}$, with $\mathcal{X}$ denoting a finite set of quantum
measurements, and she gets a classical output $a \in\mathcal{A}$, with
$\mathcal{A}$ denoting a finite set of measurement outcomes. An
\textit{assemblage} consists of the state of Bob's subsystem and the
conditional probability of Alice's outcome $a$ (correlated with Bob's state)
given the measurement choice $x$. This is specified as
\begin{equation}
\{ p_{\overline{A}
|X}(a|x),\rho_{B}^{a,x}\} _{a\in\mathcal{A},x\in\mathcal{X}}.
\end{equation}
The
sub-normalized state possessed by Bob is
\begin{equation}
\hat{\rho}_{B}^{a,x}:=p_{\overline{A}%
|X}(a|x)\rho_{B}^{a,x}.
\end{equation}
 Taking
$p_{X}(x)$ as a probability distribution over measurement choices, we can then
embed the assemblage $\{ \hat{\rho}_{B}^{a,x}\}_{a,x}$ in a classical-quantum
state as follows:
\begin{equation}
\rho_{X\overline{A}B}:=\sum_{a,x}p_{X}(x) \op{x}{x}_{X}\otimes\op{a}{a}
_{\overline{A}} \otimes\hat{\rho}_{B}^{a,x},
\end{equation}
where $\{\ket{x}_{X}\}_{x}$ and $\{\ket{a}_{\overline{A}}\}_{a}$ are orthonormal bases.
Following the approach of \cite{Gallego2015}, 
we work directly with an assemblage in what follows, such that the device on Alice's side is considered as a black box, accepting a classical input $x$ and outputting a classical variable $a$ with probability $p_{\overline{A}%
|X}(a|x)$, while the quantum state of Bob's system is $\rho_{B}^{a,x}$.


Assemblages are restricted by the no-signaling principle. That is, the reduced
state of Bob's system should not depend on the input $x$ to Alice's black box
if the measurement output $a$ is not available to him: 
\begin{equation}
\sum_{a}\hat{\rho}_{B}^{a,x}=\sum_{a}\hat{\rho}_{B}^{a,x^{\prime}}
\quad\forall x,x^{\prime}\in\mathcal{X}.
\end{equation}
This is equivalent to $I(X;B)_{\rho}=0$, where 
\begin{equation}
I(X;B)_{\rho}:=H(X)_{\rho
}+H(B)_{\rho}-H(XB)_{\rho}
\end{equation}
is the mutual information of the reduced state
$\rho_{XB}=\operatorname{Tr}_{\overline{A}}(\rho_{X\overline{A}B})$.

An assemblage is \textit{unsteerable} if arises from a classical, shared
random variable $\Lambda$ in the following sense~\cite{Wiseman2006}:
\begin{equation}
\hat{\rho}_{B}^{a,x}:= \sum_{\lambda}p_{\Lambda}(\lambda)\ p_{\overline{A}|X\Lambda
}(a|x,\lambda)\ \rho_{B}^{\lambda},
\end{equation}
where $p_{\Lambda}(\lambda)$ is a probability distribution for $\Lambda$. The
above structure indicates that the correlations observed can be explained by a
classical random variable $\Lambda$, a copy of which is sent to both Alice and
Bob, who then take actions conditioned on the particular realization $\lambda$
of $\Lambda$. The set of all unsteerable assemblages is referred to as
$\operatorname{LHS}$ (short for assemblages having a ``local-hidden-state model").

We point out that the setting considered in the resource theory of steering \cite{Gallego2015},  reviewed above, is somewhat different from that in \cite{Wiseman2006}. In the original paper \cite{Wiseman2006}, steering is considered as a property of a quantum state. That is, a quantum state is considered steerable if there exists a local measurement on Alice's system that leads to correlations that cannot be explained in terms of a local-hidden-state model. The definition considered in \cite{Gallego2015} (and that which we consider here) is to work directly with an assemblage, i.e., such that Alice's share of the bipartite quantum state is embedded in the untrusted measurement device and  the entire embedding is treated as a black box with unknown internal functioning.

As discussed above, the most general free operations allowed in the context of quantum steering are 1W-LOCC. Starting with a given assemblage $\{\hat{\rho
}_{B}^{a,x}\}_{a,x}$, it is possible for Bob to perform a quantum instrument on his system, specified as the following measurement channel
acting on an input state $\sigma_{B}$:
\begin{align}
\mathcal{M}_{B\to B^{\prime}Y}(\sigma_{B}) & := \sum_y \mathcal{K}_y(\sigma_B)\otimes\op{y}{y}_{Y},\\ \mathcal{K}_y(\sigma_B) & := \sum_t K_{y,t} \sigma_{B}K_{y,t}^{\dagger}.
\end{align}
 The sum map $\sum_y \mathcal{K}_y$ is trace preserving, i.e., $\sum_{y,t} K^{\dag}_{y,t} K_{y,t} = I_{B}$ and each $K_{y,t}$ is a Kraus
operator, taking a vector in $\mathcal{H}_{B}$ to a vector in $\mathcal{H}%
_{B^{\prime}}$. Bob can then communicate the classical result $y$ to Alice,
who chooses the input $x$ to her black box according to a classical channel
$p_{X|Y}(x|y)$. 
The state after these operations is
\begin{multline}
\rho_{X\overline{A}B^{\prime}Y}:=\sum_{a,x,y}p_{X|Y}(x|y)\op{x}{x}_{X}%
\otimes\op{a}{a}_{\overline{A}}\\
\otimes \mathcal{K}_y(\hat{\rho}_B^{a,x}) \otimes\op{y}{y}_{Y}.%
\end{multline}
Figure~\ref{1W-LOCCf} depicts a 1W-LOCC operation acting on an assemblage realized by an underlying quantum state $\rho_{AB}$ and measurement apparatus $\{M_x^a\}_a$.

\begin{figure}
\begin{center}
\includegraphics[
width=\linewidth
]{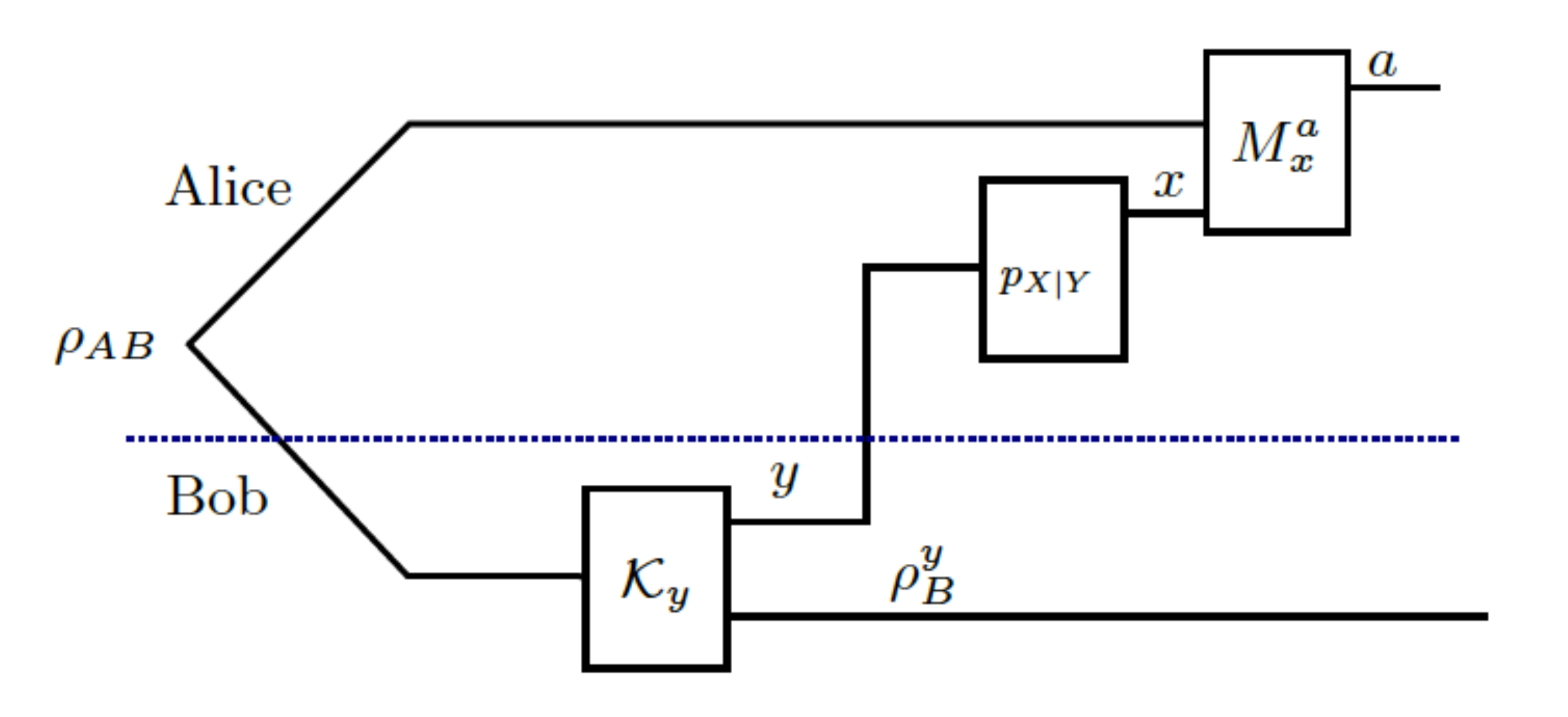}
\end{center}
\caption{This figure represents a 1W-LOCC operation acting on an assemblage realized by an underlying quantum state $\rho_{AB}$ and measurement apparatus $\{M_x^a\}_a$. Bob is allowed to send classical information $y$ to Alice, who chooses the input $x$ to her black box according to $p_{X|Y}$.}
\label{1W-LOCCf}
\end{figure}

\section{Definitions and Summary of Results}

\subsection{Intrinsic Steerability}

We allow Alice and Bob to operate on the assemblage $\{\hat{\rho}_{B}%
^{a,x}\}_{a,x}$ to maximize their correlations, resulting in the
following definition for intrinsic steerability:
\begin{definition}
[Intrinsic Steerability]\label{def:steering-CMI} Let $\{\hat{\rho}_{B}%
^{a,x}\}_{a,x}$ denote an assemblage, and let $\rho_{X\overline{A}B^{\prime}Y}$ be
a state resulting from a 1W-LOCC operation as described above. Consider a
non-signaling extension $\rho_{X\overline{A}B^{\prime}EY}$ of $\rho_{X\overline
{A}B^{\prime}Y}$ of the following form:
\begin{multline}
\rho_{X\overline{A}B^{\prime}EY}:=\sum_{x,a,y}p_{X|Y}(x|y)\op{ x}{x}_{X}%
\otimes\op{a}{ a}_{\overline{A}}\\
\otimes\hat{\rho}_{B^{\prime}E}^{a,x,y}\otimes\op{y}{y}_{Y},
\label{eq:ext-form}%
\end{multline}
where $\hat{\rho}_{B^{\prime}E}^{a,x,y}$ satisfies
\begin{equation}
\operatorname{Tr}_{E}
 (\hat{\rho}_{B^{\prime}E}^{a,x,y})= \mathcal{K}_y (\hat{\rho}_B^{a,x} )
 \end{equation}
  and
the following no-signaling constraints: %
\begin{equation}
\sum_{a}\hat{\rho}_{B^{\prime}E}^{a,x,y}=\sum_{a}\hat{\rho}_{B^{\prime}
E}^{a,x^{\prime},y} \ \forall x,x^{\prime}\in\mathcal{X},\ y\in
\mathcal{Y}. %
\end{equation}
 We define the intrinsic steerability of a given assemblage as follows:%
\begin{equation}
S(\overline{A};B)_{\hat{\rho}}:=\sup_{\left\{  p_{X|Y},\left\{  \mathcal{K}_y\right\}
_{y}\right\}  }\inf_{\rho_{X\overline{A}B^{\prime}EY}}I(X\overline{A};B^{\prime
}|EY)_{\rho},
\end{equation}
where the supremum is with respect to all quantum instruments, consisting of trace non-increasing maps $\{ \mathcal{K}_y\}_{y}$ such that the sum map $\sum_y \mathcal{K}_y$ is trace preserving and all classical channels $p_{X|Y}$ leading to Alice's input
choice~$x$. The infimum is with respect to all non-signaling extensions of
$\rho_{X\overline{A}B^{\prime}Y}$. 
Using the no-signaling constraints, which imply that $I(X;B'|EY)_{\rho}=0$, we can write
\begin{equation}
S(\overline{A};B)_{\hat{\rho}}:=\sup_{\left\{  p_{X|Y},\left\{   \mathcal{K}_y\right\}
_{y}\right\}  }\inf_{\rho_{X\overline{A}B^{\prime}EY}}I(\overline{A};B^{\prime
}|EXY)_{\rho}.\label{eq:alt-IS}
\end{equation}
\end{definition}

The idea behind the intrinsic steerability is to measure the correlations between
Alice and Bob's systems after conditioning on all of the systems that an
eavesdropper could have, with the worst possible scenario being that the
eavesdropper possesses an arbitrary non-signaling extension of $\mathcal{K}_y(\hat{\rho}_B^{a,x})$. 
We take the order of
optimizations to be similar to the order given for the squashed entanglement
of a quantum channel \cite{TGW14IEEE}: Alice and Bob first pick a 1W-LOCC
strategy to maximize their correlations, and Eve is allowed to react to this
strategy, with the goal of minimizing their correlations. Here the only restriction on Eve's system is that it has to be no-signaling. It is possible to have other restrictions on Eve's system and have modifications of the measure accordingly. Our
most fundamental result is the following theorem about intrinsic steerability.

\begin{theorem}
The intrinsic steerability $S(\overline{A};B)_{\hat{\rho}}$ is a convex steering monotone.
That is, it does not increase on average under deterministic 1W-LOCC, it
vanishes for an assemblage having a local-hidden-state model, and it is
convex. \label{monotone}
\end{theorem}

Our proof of Theorem \ref{monotone} is given in Section \ref{sec:intrinsic}.

\subsection{Restricted Intrinsic Steerability}

Definition~\ref{def:steering-CMI} might seem rather complicated with the
number of systems involved and the number of objects involved in the
optimizations. While undesirable, we note that other steering quantifiers,
such as the relative entropy of steering \cite{Gallego2015,Kaur2016}, feature similar
complications, and this seems unavoidable, having to do with the structure of
assemblages and 1W-LOCC operations.

We are thus motivated to find simpler definitions, and we can do so by
considering restricted 1W-LOCC operations as discussed above.

\begin{definition}
[Restricted Intrinsic Steerability]\label{def:reducedsteering-CMI} Let $\{\hat{\rho}_{B}^{a,x}\}_{a,x}$ denote an assemblage, and let $\rho_{X\overline
{A}B}$ denote a corresponding classical--quantum state. 
Consider a non-signaling extension
$\rho_{X\overline{A}BE}$ of $\rho_{X\overline{A}B}$ of the following form: %
\begin{equation}
\rho_{X\overline{A}B^{\prime}E}:=\sum_{a,x}p_{X}(x)\op{ x}{x}_{X}\otimes
\op{ a}{ a}_{\overline{A}}\otimes\hat{\rho}_{BE}^{a,x}, \label{eq:rext-form}%
\end{equation}
where $\hat{\rho}_{BE}^{a,x}$ satisfies $\operatorname{Tr}_{E}(\hat{\rho}%
_{BE}^{a,x})=\hat{\rho}_{B}^{a,x}$ and the following no-signaling constraints: %
\begin{equation}
\sum_{a}\hat{\rho}_{BE}^{a,x}=\sum_{a}\hat{\rho}_{BE}^{a,x^{\prime}%
} \ \forall x,x^{\prime}\in\mathcal{X} . \label{eq:no-sig-extension-RIS}%
\end{equation}
We define the restricted intrinsic steerability of $\{\hat{\rho}_{B}%
^{a,x}\}_{a,x}$ as follows:%
\begin{equation}
S^{R}(\overline{A};B)_{\hat{\rho}}:=\sup_{p_{X}}\inf_{\rho_{X\overline{A}BE}}I(X\overline
{A};B|E)_{\rho},
\end{equation}
where the supremum is with respect to all probability distributions $p_{X}$
and the infimum is with respect to all non-signaling extensions of
$\rho_{X\overline{A}B}$. 
Using the no-signaling constraints, which imply that $I(X;B|E)_{\rho}=0$, it follows that 
\begin{equation}
S^{R}(\overline{A};B)_{\hat{\rho}}:=\sup_{p_{X}}\inf_{\rho_{X\overline{A}BE}}I(\overline
{A};B|EX)_{\rho}.\label{eq:alt-RIS}
\end{equation}
\end{definition}

We prove
that the restricted intrinsic steerability is a steering monotone with respect
to restricted 1W-LOCC and that it is
convex. 

\begin{theorem}
\label{thm:RIS-monotone} The restricted intrinsic steerability $S^{R}(\overline
{A};B)_{\hat{\rho}}$ is a convex steering monotone with respect to restricted
1W-LOCC. That is, it does not increase under restricted
deterministic 1W-LOCC, it vanishes for assemblages having a local-hidden-state
model, and it is
convex.
\end{theorem}

Our proof for Theorem \ref{thm:RIS-monotone} is given in Section \ref{sec:rest}.
\bigskip 

 By inspecting definitions, we can conclude that intrinsic steerability is
never smaller than restricted intrinsic steerability: %
\begin{equation}
S(\overline{A};B)_{\hat{\rho}}\geq S^{R}(\overline{A};B)_{\hat{\rho}}%
. \label{eq:IS-to-RIS}%
\end{equation}
This follows because the restricted intrinsic steerability involves a
supremization over particular 1W-LOCC\ strategies that are included in the
supremization in the definition of the intrinsic steerability.


By using known bounds on conditional mutual information, the expression in \eqref{eq:alt-IS}, and the fact that taking an infimum over classical extensions
$E$ does not decrease $S(\overline{A};B)_{\hat{\rho}}$, we can conclude that
\begin{equation}
0\leq S(\overline{A};B)_{\hat{\rho}}\leq\log_{2}|\overline{A}|.
\end{equation}
 The lower bound
follows from the strong subadditivity of quantum entropy \cite{LR73}\ and the
upper bound follows from a dimension bound (see, e.g., \cite{W15book}). Similarly, using known bounds on conditional mutual information, the expression in \eqref{eq:alt-RIS}, and the fact that taking an infimum over classical extensions
$E$ does not decrease $S^{R}(\overline{A};B)_{\hat{\rho}}$, we find that
\begin{equation}
0\leq
S^{R}(\overline{A};B)_{\hat{\rho}}\leq\min\{\log_{2}|\overline{A}|,\log_{2}|B|\}.
\end{equation}

\section{Examples\label{sec:example}}

As an example, consider the following ``BB84 assemblage'' resulting from Pauli
$\sigma_{Z}$ or $\sigma_{X}$ measurements on one share of a maximally entangled state
\begin{equation}
|\Phi\rangle_{AB}:=(|00\rangle_{AB}+|11\rangle_{AB})/\sqrt{2},
\end{equation}
consisting of
the following four subnormalized states: %
\begin{align}
\hat{\rho}_{B}^{a=0,x=0} &    =\tfrac{1}{2}|0\rangle\langle0|_{B},  \\
\hat{\rho
}_{B}^{a=1,x=0} &  =\tfrac{1}{2}|1\rangle\langle1|_{B},\\
\hat{\rho}_{B}^{a=0,x=1} &   =\tfrac{1}{2}|+\rangle\langle+|_{B}, \\
 \hat{\rho
}_{B}^{a=1,x=1} & =\tfrac{1}{2}|-\rangle\langle-|_{B}.
\end{align}
As we show in the proof of Proposition~\ref{ex:simple}, the non-signaling constraint for this case imposes that any non-signaling
extension of the above assemblage has the form $\hat{\rho}_{B}^{a,x}%
\otimes\omega_{E}$ for all $a,x\in\{0,1\}$ and for some state $\omega_{E}$. \textbf{\textit{Thus, in
this sense, the BB84 assemblage is unextendible and features a certain kind of monogamy against non-signaling adversaries}}. As a consequence, we find that
this assemblage has exactly one bit of intrinsic steerability.

In Proposition~\ref{ex:gen}, we generalize the above result to an assemblage resulting from an arbitrary pure bipartite state being measured in the Schmidt basis and the basis Fourier conjugate to this one. We find that this assemblage has the same kind of monogamy against non-signaling adversaries and that it has restricted intrinsic steerability equal to the entropy of entanglement \cite{BBPS96} of the state being measured.


\begin{proposition}
\label{ex:simple}
Consider a maximally entangled state%
\begin{equation}
|\Phi\rangle_{AB}:=\frac{1}{\sqrt{2}}(|00\rangle_{AB}+|11\rangle_{AB}).
\end{equation}
Let measurement $x=0$ be Pauli $\sigma_{Z}$ on system $A$, with outcomes $a=0$
and $a=1$. Let measurement $x=1$ be Pauli $\sigma_{X}$ on system $A$, with
outcomes $a=0$ and $a=1$. This leads to the following assemblage:%
\begin{align}
\left\{\begin{array}{lr}\hat{\rho}_{B}^{a=0,x=0}=\frac{1}{2}|0\rangle\langle0|_{B}
,\ &\hat{\rho}_{B}^{a=1,x=0}=\frac{1}{2}|1\rangle\langle1|_{B},\\ \hat{\rho}%
_{B}^{a=0,x=1}=\frac{1}{2}|+\rangle\langle+|_{B},\ &\hat{\rho}_{B}%
^{a=1,x=1}=\frac{1}{2}|-\rangle\langle-|_{B}\end{array}\right\},
\end{align}
which has one bit of intrinsic steerability and restricted intrinsic steerability:
\begin{equation}
S(\overline
{A};B)_{\hat{\rho}} = S^{R}(\overline
{A};B)_{\hat{\rho}} = 1.
\end{equation}
\end{proposition}

\begin{proof}
Arbitrary extensions of each of the above subnormalized states are as follows:%
\begin{align}
\hat{\rho}_{BE}^{a=0,x=0}  &  =\frac{1}{2}|0\rangle\langle0|_{B}\otimes
\omega_{E}^{00},\\
\hat{\rho}_{BE}^{a=1,x=0}
&=\frac{1}{2}|1\rangle
\langle1|_{B}\otimes\omega_{E}^{10},\\
\hat{\rho}_{BE}^{a=0,x=1}  &  =\frac{1}{2}|+\rangle\langle+|_{B}\otimes
\omega_{E}^{01},\\
\hat{\rho}_{BE}^{a=1,x=1} & =\frac{1}{2}|-\rangle
\langle-|_{B}\otimes\omega_{E}^{11},
\end{align}
where $\omega_{E}^{ij}\geq0$ and $\operatorname{Tr}(\omega_{E}^{ij})=1$ for
all $i,j\in\{0,1\}$. The no-signaling constraint is as follows:%
\begin{equation}
\hat{\rho}_{BE}^{a=0,x=0}+\hat{\rho}_{BE}^{a=1,x=0}=\hat{\rho}_{BE}%
^{a=0,x=1}+\hat{\rho}_{BE}^{a=1,x=1}. \label{eq:example-no-sig-constr}%
\end{equation}
Writing out the left-hand side of \eqref{eq:example-no-sig-constr}\ in matrix
form, we find that%
\begin{equation}
\frac{1}{2}|0\rangle\langle0|_{B}\otimes\omega_{E}^{00}+\frac{1}{2}%
|1\rangle\langle1|_{B}\otimes\omega_{E}^{10}=\frac{1}{2}%
\begin{bmatrix}
\omega_{E}^{00} & 0\\
0 & \omega_{E}^{10}%
\end{bmatrix}
.
\end{equation}
Writing out the right-hand side of \eqref{eq:example-no-sig-constr} in matrix
form, we find that
\begin{align}
&\frac{1}{2}|+\rangle\langle+|_{B}\otimes\omega_{E}^{01}+\frac{1}{2}%
|-\rangle\langle-|_{B}\otimes\omega_{E}^{11}  \\
&  =\frac{1}{4}\left[
|0\rangle\langle0|_{B}+|1\rangle\langle0|_{B}+|0\rangle\langle1|_{B}%
+|1\rangle\langle1|_{B}\right]  \otimes\omega_{E}^{01}\nonumber\\
& \quad +\frac{1}{4}\left[  |0\rangle\langle0|_{B}-|1\rangle\langle
0|_{B}-|0\rangle\langle1|_{B}+|1\rangle\langle1|_{B}\right]  \otimes\omega
_{E}^{11}\\
&  =\frac{1}{2}|0\rangle\langle0|_{B}\otimes\left(  \frac{\omega_{E}%
^{01}+\omega_{E}^{11}}{2}\right)  \nonumber\\
& \qquad +\frac{1}{2}|1\rangle\langle0|_{B}%
\otimes\left(  \frac{\omega_{E}^{01}-\omega_{E}^{11}}{2}\right) \nonumber\\
& \qquad+\frac{1}{2}|0\rangle\langle1|_{B}\otimes\left(  \frac{\omega
_{E}^{01}-\omega_{E}^{11}}{2}\right)  \nonumber\\
& \qquad+\frac{1}{2}|1\rangle\langle
1|_{B}\otimes\left(  \frac{\omega_{E}^{01}+\omega_{E}^{11}}{2}\right). \\
&  =\frac{1}{2}%
\begin{bmatrix}
\frac{\omega_{E}^{01}+\omega_{E}^{11}}{2} & \frac{\omega_{E}^{01}-\omega
_{E}^{11}}{2}\\
\frac{\omega_{E}^{01}-\omega_{E}^{11}}{2} & \frac{\omega_{E}^{01}+\omega
_{E}^{11}}{2}%
\end{bmatrix}.
\end{align}

So equating them, we find that the following equation (no-signaling
constraint)\ should be satisfied%
\begin{equation}%
\begin{bmatrix}
\omega_{E}^{00} & 0\\
0 & \omega_{E}^{10}%
\end{bmatrix}
=%
\begin{bmatrix}
\frac{\omega_{E}^{01}+\omega_{E}^{11}}{2} & \frac{\omega_{E}^{01}-\omega
_{E}^{11}}{2}\\
\frac{\omega_{E}^{01}-\omega_{E}^{11}}{2} & \frac{\omega_{E}^{01}+\omega
_{E}^{11}}{2}%
\end{bmatrix}
.
\end{equation}
This implies that $\omega_{E}^{01}=\omega_{E}^{11}$, which in turn implies
that $\omega_{E}^{10}=\omega_{E}^{01}=\omega_{E}^{11}=\omega_{E}^{00}$. Thus,
the only possible extension allowed in order to satisfy the no-signaling
constraint is a product extension independent of $a$ and $x$, meaning one of
the following form:%
\begin{align}
\hat{\rho}_{BE}^{a=0,x=0} & =\frac{1}{2}|0\rangle\langle0|_{B}\otimes
\omega_{E}, \nonumber \\
 \hat{\rho}_{BE}^{a=1,x=0}& =\frac{1}{2}|1\rangle\langle
1|_{B}\otimes\omega_{E},\nonumber\\
\hat{\rho}_{BE}^{a=0,x=1}  & =\frac{1}{2}|+\rangle\langle+|_{B}\otimes
\omega_{E},\nonumber \\
\hat{\rho}_{BE}^{a=1,x=1}& =\frac{1}{2}|-\rangle\langle
-|_{B}\otimes\omega_{E},
\end{align}
where $\omega_{E}\geq0$ and $\operatorname{Tr}(\omega_{E})=1$. We can then
evaluate the restricted intrinsic steerability in terms of the following
classical--quantum state:%
\begin{equation}
\left[
\begin{array}
[c]{c}%
\frac{1}{2}|0\rangle\langle0|_{X}\otimes|0\rangle\langle0|_{\overline{A}}%
\otimes\frac{1}{2}|0\rangle\langle0|_{B}
\\+\frac{1}{2}|0\rangle\langle
0|_{X}\otimes|1\rangle\langle1|_{\overline{A}}\otimes\frac{1}{2}|1\rangle
\langle1|_{B}\\
+\frac{1}{2}|1\rangle\langle1|_{X}\otimes|0\rangle\langle0|_{\overline{A}}%
\otimes\frac{1}{2}|
+\rangle\langle+|_{B}\\+\frac{1}{2}|1\rangle\langle
1|_{X}\otimes|1\rangle\langle1|_{\overline{A}}\otimes\frac{1}{2}|-\rangle
\langle-|_{B}%
\end{array}
\right]  \otimes\omega_{E}.
\end{equation}
The conditional mutual information of this state is as follows:
\begin{align}
I(X\overline{A};B|E)&=I(X\overline{A};B)\nonumber\\&=H(B)-H(B|X\overline{A})=H(B)=1,
\end{align}
so that this assemblage has \textit{one bit of restricted intrinsic
steerability}. The first equality follows because the system $E$ is product
regardless of the extension, due to the above analysis with the no-signaling
constraint. The second equality follows by expanding the mutual information.
The third equality follows because the state of the $B$ system is pure when
conditioned on systems $X\overline{A}$. The final equality follows because the
reduced state on the $B$ system is maximally mixed. Also, it is clear that
this is the maximum value of the restricted intrinsic steerability, given that
it is always bounded from above by $\log\dim(\mathcal{H}_{B})$
or $\log\dim(\mathcal{H}_{\overline{A}})$. By considering the upper bound $\log\dim(\mathcal{H}_{\overline{A}})$ for intrinsic steerability, we see that this assemblage achieves the upper bound on 
intrinsic steerability and thus has one bit of intrinsic steerability.
\end{proof}

\begin{proposition}
\label{ex:gen}
Consider a pure bipartite state $|\varphi\rangle_{AB}$ in its Schmidt basis:%
\begin{equation}
|\varphi\rangle_{AB}:=\sum_{j=0}^{d-1}\alpha_{j}|j\rangle_{A}\otimes
|j\rangle_{B},
\end{equation}
where $\left\vert \alpha_{j}\right\vert \neq0$ for all $j\in\{0,\ldots,d-1\}$.
Let measurement $x=0$ be a measurement $\{|j\rangle\langle j|_{A}\}_{j}$\ in
the Schmidt basis on system $A$, with outcomes $a=j\in\{0,\ldots,d-1\}$. Let
measurement $x=1$ be a measurement $\{|\widetilde{j}\rangle\langle
\widetilde{j}|_{A}\}_{j}$ in the Fourier conjugate basis, where%
\begin{equation}
|\widetilde{j}\rangle_{A}:=\frac{1}{\sqrt{d}}\sum_{k}e^{2\pi ijk/d}%
|k\rangle_{A},
\end{equation}
on system $A$, with outcomes $a=j\in\{0,\ldots,d-1\}$. This leads to the
following assemblage:
\begin{multline}
\Bigg\{  \left\{  \hat{\rho}_{B}^{a=j,x=0}=\left\vert \alpha_{j}\right\vert
^{2}|j\rangle\langle j|_{B}\right\}  _{j},\\ \left\{  \hat{\rho}_{B}%
^{a=j,x=1}=\frac{1}{d}Z^{\dag}(j)|\psi\rangle\langle\psi|_{B}Z(j)\right\}
_{j},\Bigg\}  ,
\end{multline}
where $|\psi\rangle_{B}:=\sum_{j}\alpha_{j}|j\rangle_{B}$. This assemblage has
\begin{equation}
H(\{\left\vert \alpha_{j}\right\vert ^{2}\}_{j})=H(A)_{\varphi}
\end{equation}
 bits of
restricted intrinsic steerability. Note that this is equal to the entropy of
entanglement of the state $|\varphi\rangle_{AB}$. If the state 
$|\varphi\rangle_{AB}$ is maximally entangled so that $\alpha_j = 1/\sqrt{d}$, then the resulting assemblage has $\log_2(d)$ bits of intrinsic steerability.
\end{proposition}

\begin{proof}
It is clear that the post-measurement state for Bob $\hat{\rho}_{B}^{a=j,x=0}$
is as above. For the other case, consider that%
\begin{align}
& \!\!\!\!\langle\widetilde{j}|_{A}\otimes I_{B}|\varphi\rangle_{AB} \nonumber \\
&  =\frac{1}%
{\sqrt{d}}\sum_{k=0}^{d-1}e^{-2\pi ijk/d}\langle k|_{A}\sum_{l=0}^{d-1}%
\alpha_{l}|l\rangle_{A}\otimes|l\rangle_{B}\\
&  =\frac{1}{\sqrt{d}}\sum_{k,l=0}^{d-1}\alpha_{k}e^{-2\pi ijk/d}\langle
k|l\rangle_{A}\otimes|l\rangle_{B}\\
&  =\frac{1}{\sqrt{d}}\sum_{k=0}^{d-1}\alpha_{k}e^{-2\pi ijk/d}|k\rangle_{B}.
\end{align}
Now defining the unitary operator $Z(j)$ by $Z(j)|k\rangle=e^{2\pi
ijk/d}|k\rangle$ for $j,k\in\{0,\ldots,d-1\}$, we can write%
\begin{equation}
\langle\widetilde{j}|_{A}\otimes I_{B}|\varphi\rangle_{AB}=\frac{1}{\sqrt{d}%
}Z^{\dag}(j)|\psi\rangle_{B},
\end{equation}
confirming the post-measurement subnormalized states $\hat{\rho}_{B}%
^{a=j,x=1}$. Arbitrary extensions of each of the above subnormalized states
are as follows:%
\begin{align}
\hat{\rho}_{BE}^{a=j,x=0} &  =\left\vert \alpha_{j}\right\vert ^{2}%
|j\rangle\langle j|_{B}\otimes\omega_{E}^{j},\\
\hat{\rho}_{BE}^{a=j,x=1} &  =\frac{1}{d}Z^{\dag}(j)|\psi\rangle\langle
\psi|_{B}Z(j)\otimes\tau_{E}^{j},
\end{align}
where $\omega_{E}^{j},\tau_{E}^{j}\geq0$ and $\operatorname{Tr}(\omega_{E}%
^{j})=\operatorname{Tr}(\tau_{E}^{j})=1$ for all $j\in\{0,\ldots,d-1\}$. The
no-signaling constraint is as follows:%
\begin{equation}
\sum_{j=0}^{d-1}\hat{\rho}_{BE}^{a=j,x=0}=\sum_{j=0}^{d-1}\hat{\rho}%
_{BE}^{a=j,x=1},
\end{equation}
which is the same as%
\begin{align}
&\sum_{k=0}^{d-1}|k\rangle\langle k|_{B}\otimes\left\vert \alpha_{k}\right\vert
^{2}\omega_{E}^{k} \nonumber 
\\ &  =\sum_{j=0}^{d-1}\frac{1}{d}Z^{\dag}(j)|\psi
\rangle\langle\psi|_{B}Z(j)\otimes\tau_{E}^{j}\label{eq:no-sig-gen-constr}\\
&  =\sum_{j,k,k^{\prime}=0}^{d-1}\frac{1}{d}\alpha_{k}\alpha_{k^{\prime}%
}^{\ast}e^{-2\pi ij(k-k^{\prime})/d}|k\rangle\langle k^{\prime}|_{B}%
\otimes\tau_{E}^{j}\\
&  =\sum_{k,k^{\prime}=0}^{d-1}|k\rangle\langle k^{\prime}|_{B}\otimes\frac
{1}{d}\alpha_{k}\alpha_{k^{\prime}}^{\ast}\sum_{j=0}^{d-1}e^{-2\pi
ij(k-k^{\prime})/d}\tau_{E}^{j}.
\end{align}
Set $k^{\prime}=0$. For $k\in\{0,1,\ldots,d-1\}$, we get the following
constraints from the no-signaling condition:%
\begin{align}
\omega_{E}^{0} & =\frac{1}{d}\sum_{j=0}^{d-1}\tau_{E}^{j}%
,\label{eq:tau-equations-NS}\\
0 & = \sum_{j=0}^{d-1}e^{-2\pi ijk/d}\tau_{E}^{j}  
.\label{eq:tau-equations-NS-2}%
\end{align}
We can conclude that $\tau_{E}^{j}$ is independent of $j$, so that $\tau
_{E}^{j}=\omega_{E}^{0}$ for all $j\in\{0,\ldots,d-1\}$. To see this, let us
solve the above equations, thinking of $\omega_{E}^{0}$ as fixed and $\tau
_{E}^{j}$ as free for all $j\in\{0,\ldots,d-1\}$. Consider that%
\begin{equation}
\sum_{j=0}^{d-1}e^{-2\pi ijk/d}=0\ \ \ \forall k\in\{1,\ldots,d-1\}.
\end{equation}
Then we can see that $\tau_{E}^{0}=\tau_{E}^{1}=\cdots=\tau_{E}^{d-1}%
=\omega_{E}^{0}$ is one of the solutions of the equations in
\eqref{eq:tau-equations-NS}--\eqref{eq:tau-equations-NS-2}. Since the
equations are linearly independent, it is a unique solution. Now considering
the other blocks in \eqref{eq:no-sig-gen-constr}\ (i.e., for $k=k^{\prime
}=1,\ldots,d-1$), we find that $\omega_{E}^{1}=\cdots=\omega_{E}^{d-1}%
=\omega_{E}^{0}$. Thus, the only possible extension allowed in order to
satisfy the no-signaling constraint is a product extension independent of $a$
and $x$, meaning one of the following form:%
\begin{align}
\hat{\rho}_{BE}^{a=j,x=0} &  =\left\vert \alpha_{j}\right\vert ^{2}%
|j\rangle\langle j|_{B}\otimes\omega_{E},\\
\hat{\rho}_{BE}^{a=j,x=1} &  =\frac{1}{d}Z^{\dag}(j)|\psi\rangle\langle
\psi|_{B}Z(j)\otimes\omega_{E},
\end{align}
where $\omega_{E}\geq0$ and $\operatorname{Tr}(\omega_{E})=1$. We can then
evaluate the restricted intrinsic steerability in terms of the following
classical--quantum state:%
\begin{multline}
\Bigg[  p|0\rangle\langle0|_{X}\otimes\sum_{j}|j\rangle\langle j|_{\overline{A}%
}\otimes\left\vert \alpha_{j}\right\vert ^{2}|j\rangle\langle j|_{B}\\
+\left(
1-p\right)  |1\rangle\langle1|_{X}\otimes\sum_{j}|j\rangle\langle j|_{\overline{A}%
}\\\otimes\frac{1}{d}Z^{\dag}(j)|\psi\rangle\langle\psi|_{B}Z(j)\Bigg]
\otimes\omega_{E},
\end{multline}
where $(p,1-p)$ is a probability distribution for the input $x$. The
conditional mutual information of this state is as follows:%
\begin{align}
I(X\overline{A};B|E)   &=I(X\overline{A};B)
  =H(B)-H(B|X\overline{A})\\
  &=H(B)
  =H(\{\left\vert \alpha_{j}\right\vert ^{2}\}),
\end{align}
so that this assemblage has
$H(\{\left\vert \alpha_{j}\right\vert ^{2}%
\})$
\textit{ bits of restricted intrinsic steerability}. The first step
follows because the system $E$ is product regardless of the extension, due to
the above analysis with the no-signaling constraint. The second step follows
by expanding the mutual information. The third step follows because the state
of the $B$ system is pure when conditioned on systems $X\overline{A}$. The final
step follows because the reduced state on the $B$ system is $\sum
_{j}\left\vert \alpha_{j}\right\vert ^{2}|j\rangle\langle j|_{B}$, which can
be seen from
\begin{align}
& \operatorname{Tr}_{X\overline{A}}\Bigg(  p|0\rangle\langle0|_{X}\otimes\sum
_{j}|j\rangle\langle j|_{\overline{A}}\otimes\left\vert \alpha_{j}\right\vert
^{2}|j\rangle\langle j|_{B}  \nonumber \\
& \qquad \qquad +
\left(  1-p\right)  |1\rangle\langle1|_{X}%
\otimes\sum_{j}|j\rangle\langle j|_{\overline{A}}\nonumber 
\\
& \qquad \qquad \otimes\frac{1}{d}Z^{\dag}%
(j)|\psi\rangle\langle\psi|_{B}Z(j)\Bigg)  \nonumber\\
& =p\sum_{j}\left\vert \alpha_{j}\right\vert ^{2}|j\rangle\langle
j|_{B}+\left(  1-p\right)  \sum_{j}\frac{1}{d}Z^{\dag}(j)|\psi\rangle
\langle\psi|_{B}Z(j)\\
& =p\sum_{j}\left\vert \alpha_{j}\right\vert ^{2}|j\rangle\langle
j|_{B}+\left(  1-p\right)  \sum_{j}\left\vert \alpha_{j}\right\vert
^{2}|j\rangle\langle j|_{B}\\
& =\sum_{j}\left\vert \alpha_{j}\right\vert ^{2}|j\rangle\langle j|_{B}.
\end{align}
This state is independent of the input probability distribution, so that the
maximum is achieved for any choice of $p\in(0,1)$.

If the state $|\varphi\rangle_{AB}$ is maximally entangled, then
\begin{equation}
H(\{\left\vert \alpha_{j}\right\vert ^{2}%
\}) = \log_2(d).
\end{equation}
Given the upper bound
$\log(\dim(\mathcal{H}_{\overline{A}})) = \log_2(d)$ on intrinsic steerability, we see that the upper bound is achieved in this case.
\end{proof}


\section{Intrinsic Steerability}
\label{sec:intrinsic}

We now give a proof for Theorem~\ref{monotone} and proofs for other properties of the intrinsic
steerability stated earlier.

\begin{proposition}
Intrinsic steerability vanishes for assemblages having an LHS model.
\end{proposition}

\begin{proof}
To prove this, consider the following particular non-signaling extension 
for an assemblage with a local-hidden-state model:
\begin{multline}
\sum_{x,a,\lambda,y}p_{X|Y}(x|y)\op{ x}{x}_{X}\otimes p_{\overline{A}|X\Lambda
}(a|x,\lambda)\op{a}{ a}_{\overline{A}}
\\
\otimes \sum_t K_{y,t}\hat{\rho}_{B}^{\lambda}K_{y,t}^{\dag}\otimes p_{\Lambda}%
(\lambda)\op{\lambda}{\lambda}_{E}\otimes\op{y}{y}_{Y}.
\end{multline}
For this non-signaling extension, conditioned on the values $\lambda$ and $y$,
systems $X\overline{A}$ and $B^{\prime}$ are in a product state, so that the
conditional mutual information $I(X\overline{A};B^{\prime}|EY)$ vanishes. The same
argument applies to all quantum instruments $\{\mathcal{K}_y\}_{y}$ and channels $p_{X|Y}$,
so that
\begin{equation}
S(\overline{A};B)_{\rho}=0
\end{equation}
in this case.
\end{proof}

\begin{proposition}
[1W-LOCC monotone]Let $\{\hat{\rho}_{B}^{a,x}\}_{a,x}$ be an assemblage, and suppose that%
\begin{multline}
\Bigg\{  \hat{\rho}_{B_{f},z}^{a_{f},x_{f}}:=\\
\qquad\sum_{a,x}p(a_{f}|x_{f}%
,x,a,z)p(x|x_{f},z)\mathcal{K}_z(\hat{\rho}_{B}^{a,x})/p(z)\Bigg\}
_{a_{f},x_{f}},
\end{multline}
is an assemblage that arises from it by the action of a general
1W-LOCC\ operation, where%
\begin{equation}
p(z):=\operatorname{Tr}\!\left( \mathcal{K}_z\!\left(\sum_{a}\hat{\rho}_{B}%
^{a,x}\right)\right)  =\operatorname{Tr}(\mathcal{K}_z(\rho_{B})).
\end{equation}
Then the intrinsic steerability is monotone on average under deterministic 1W-LOCC, in the following sense:
\begin{equation}
\sum_{z}p(z)S(\overline{A}_{f};B_{f})_{\hat{\rho}_{z}}\leq S(\overline{A};B)_{\hat{\rho
}}. \label{eq:1W-LOCC-monotone}%
\end{equation}
\end{proposition}

\begin{proof}
First, we give a proof sketch for the
monotonicity of intrinsic steerability on average under deterministic 1W-LOCC: %
\begin{equation}
S(\overline{A};B)_{\hat{\rho}}\geq\sum_{z}p_{Z}(z)S(\overline{A}_{f};B_{f})_{\hat{\rho
}_{z}},
\end{equation}
where $\hat{\rho}_{z}:=\{\hat{\rho}_{B_{f},z}^{a_{f},x_{f}}\}_{a_{f},x_{f}}$
is the assemblage resulting from a 1W-LOCC operation on the initial assemblage
$\{\hat{\rho}_{B}^{a,x}\}_{a,x}$ and is given as \cite{Gallego2015} %
\begin{equation}
\hat{\rho}_{B_{f},z}^{a_{f},x_{f}}:=\sum_{a,x}p(a_{f}|a,x,x_{f},z)p(x|x_{f}%
,z)\mathcal{K}_z(\hat{\rho}_{B}^{a,x}).
\end{equation}
In the above, $p(a_{f}|a,x,x_{f},z)$ and $p(x|x_{f},z)$ are local classical
channels that Alice uses, respectively, to pick the output $a_{f}$ of the
final assemblage and the input $x$ to her initial assemblage. The set
$\{\mathcal{K}_z\}_{z}$ is such that the sum map $\sum_z\mathcal{K}_z$ is trace preserving and thus corresponds
to a measurement of Bob's system. The definition of the intrinsic steerability
involves a supremum over measurements of the system $B_{f}$ of the final
assemblage and classical channels for the input $X_{f}$ to the final
assemblage. Using data processing and when given $Z$, we can say that system
$\overline{A}_{f}$ was obtained by processing systems $XX_{f}\overline{A}$. Then, the
two successive measurements on Bob's system can be thought of as a single
measurement. Since the intrinsic steerability involves a supremum over all possible measurements, the result follows.

We now give a detailed proof. 
To see this, consider that, in accordance with the definition of $S(\overline
{A}_{f};B_{f})_{\hat{\rho}_{z}}$, the assemblages $\{\hat{\rho}_{B_{f}%
,z}^{a_{f},x_{f}}\}_{a_{f},x_{f}}$ can be further preprocessed by a
$z$-dependent 1W-LOCC $\{p_{X_{f}|YZ=z},\{\mathcal{L}_y^{(z)}\}_{y}\}$, resulting in
the following state:%
\begin{equation}
\sigma_{X_{f}\overline{A}_{f}B_{f}^{\prime}Y}^{z}:=\sum_{a_{f},x_{f},y}%
p(x_{f}|yz)[x_{f}]\otimes\lbrack a_{f}]\otimes \mathcal{L}_y^{(z)}(\hat{\rho}_{B_{f}%
,z}^{a_{f},x_{f}})\otimes\lbrack y].
\label{eq:smaller-state-def}%
\end{equation}

\begin{notation}
In the above and in what follows, we employ a shorthand $[x]\equiv
|x\rangle\langle x|_{X}$ or $[a]\equiv|a\rangle\langle a|_{\overline{A}}$, etc.
\end{notation}

\noindent The state in\ \eqref{eq:smaller-state-def}\ is extended by the
following one:%
\begin{multline}
\sigma_{X_{f}X\overline{A}_{f}\overline{A}B_{f}^{\prime}Y}^{z} :=\sum_{a_{f},a,x,x_{f}%
,y}p(x_{f}|yz)[x_{f}]\\
\otimes p(x|x_{f},z)[x]\otimes
 p(a_{f}|x_{f}%
,x,a,z)[a_{f}]\otimes\lbrack a]\\\otimes\frac{\mathcal{L}_y^{(z)}(\hat{\rho}%
_{B}^{a,x})}{p(z)}\otimes\lbrack y],
\end{multline}
which in turn are elements of the following classical--quantum state:%
\begin{equation}
\sigma_{X_{f}X\overline{A}_{f}\overline{A}B_{f}^{\prime}YZ}:=\sum_{z}\sigma_{X_{f}%
X\overline{A}_{f}\overline{A}B_{f}^{\prime}Y}^{z}\otimes p(z)[z].
\end{equation}
An \textit{arbitrary} non-signaling extension of the state in
\eqref{eq:smaller-state-def}, according to that needed in the definition of
$S(\overline{A}_{f};B_{f})_{\hat{\rho}_{z}}$,\ is as follows:%
\begin{multline}
\sigma_{X_{f}\overline{A}_{f}B_{f}^{\prime}EY}^{z}:=\sum_{a_{f},x_{f},y}%
p(x_{f}|yz)[x_{f}]\otimes\lbrack a_{f}]\\
\otimes\hat{\tau}_{B_{f}^{\prime}%
E}^{a_{f},x_{f},y,z}\otimes\lbrack y], \label{eq:generic-ext-lower-bnd}%
\end{multline}
where $\hat{\tau}_{B_{f}^{\prime}E}^{a_{f},x_{f},y,z}$\ satisfies%
\begin{align}
\operatorname{Tr}_{E}(\hat{\tau}_{B_{f}^{\prime}E}^{a_{f},x_{f},y,z})  &
=\mathcal{L}_y^{(z)}(\hat{\rho}_{B_{f},z}^{a_{f},x_{f}}),\\
\sum_{a_{f}}\hat{\tau}_{B_{f}^{\prime}E}^{a_{f},x_{f},y,z}  &  =\sum_{a_{f}%
}\hat{\tau}_{B_{f}^{\prime}E}^{a_{f},x_{f}^{\prime},y,z}\ \ \ \ \nonumber\\&\forall
x_{f},x_{f}^{\prime}\in\mathcal{X}_{f},\ y\in\mathcal{Y},\ z\in\mathcal{Z}.
\end{align}
A \textit{particular} non-signaling extension of the state in
\eqref{eq:smaller-state-def}, according to that needed in the definition of
$S(\overline{A}_{f};B_{f})_{\hat{\rho}_{z}}$,\ is as follows:%
\begin{multline}
\zeta_{X_{f}\overline{A}_{f}B_{f}^{\prime}EY}^{z}:=\sum_{a_{f},x_{f},y}%
p(x_{f}|yz)[x_{f}]\otimes\lbrack a_{f}]\\\otimes\sum_{a,x}p(a_{f}|x_{f}%
,x,a,z)p(x|x_{f},z)\hat{\omega}_{B_{f}^{\prime}E}^{a,x,y,z}\otimes\lbrack y],
\label{eq:particular-ext}%
\end{multline}
where $\hat{\omega}_{B_{f}^{\prime}E}^{a,x,y,z}$ satisfies%
\begin{align}
\operatorname{Tr}_{E}(\hat{\omega}_{B_{f}^{\prime}E}^{a,x,y,z})  &
=\frac{\mathcal{L}_y^{(z)}(\mathcal{K}_z(\hat{\rho}_{B}^{a,x}))%
}{p(z)},\\
\sum_{a}\hat{\omega}_{B_{f}^{\prime}E}^{a,x,y,z}  &  =\sum_{a}\hat{\omega
}_{B_{f}^{\prime}E}^{a,x,y,z}\ \ \ \ \forall x,x^{\prime}\in\mathcal{X}%
,\ y\in\mathcal{Y},\ z\in\mathcal{Z}.
\end{align}
The operator $\hat{\omega}_{B_{f}^{\prime}E}^{a,x,y,z}$ will serve as an
arbitrary non-signaling extension needed in the definition of $S(\overline
{A};B)_{\hat{\rho}}$. Let $\zeta_{X_{f}\overline{A}_{f}B_{f}^{\prime}EYZ}$ denote
the following state:%
\begin{equation}
\zeta_{X_{f}\overline{A}_{f}B_{f}^{\prime}EYZ}:=\sum_{z}\zeta_{X_{f}\overline{A}%
_{f}B_{f}^{\prime}EY}^{z}\otimes p(z)[z].
\end{equation}
This in turn is a marginal of the following state:
\begin{multline}
\zeta_{X_{f}X\overline{A}_{f}\overline{A}B_{f}^{\prime}EYZ}:=\sum_{a_{f},a,x_{f}%
,x,y}p(x_{f}|yz)[x_{f}]\\
\otimes p(x|x_{f},z)[x]
\otimes p(a_{f}|x_{f}%
,x,a,z)[a_{f}] \otimes\lbrack a]\\
\otimes\hat{\omega}_{B_{f}^{\prime}%
E}^{a,x,y,z}\otimes\lbrack y]\otimes p(z)[z].\label{eq:big-state}%
\end{multline}

Consider that%
\begin{align}
&\sum_{z}p(z)\inf_{\text{ext.~in \eqref{eq:generic-ext-lower-bnd}}}I(X_{f}%
\overline{A}_{f};B_{f}^{\prime}|EY)_{\sigma^{z}}  
\nonumber 
\\&  \leq\sum_{z}p(z)I(X_{f}\overline
{A}_{f};B_{f}^{\prime}|EY)_{\zeta^{z}}\\
&  =I(X_{f}\overline{A}_{f};B_{f}^{\prime}|EYZ)_{\zeta}\\
&  \leq I(X_{f}X\overline{A};B_{f}^{\prime}|EYZ)_{\zeta}\\
&  =I(X\overline{A};B_{f}^{\prime}|EYZ)_{\zeta}+I(X_{f};B_{f}^{\prime}|EYZX\overline
{A})_{\zeta}\\
&  =I(X\overline{A};B_{f}^{\prime}|EYZ)_{\zeta}.
\end{align}
The first inequality follows because the extension state $\zeta_{X_{f}\overline
{A}_{f}B_{f}^{\prime}EY}^{z}$ is a particular kind of non-signaling extension
required in the definition of $S(\overline{A}_{f};B_{f})_{\hat{\rho}_{z}}$. The
first equality follows because system $Z$ is classical and thus can be
incorporated as a conditioning system in the conditional mutual information.
The second inequality follows from local data processing for the conditional
mutual information:\ given $Z$, the system $\overline{A}_{f}$ arises from local
processing of systems $X_{f}X\overline{A}$. The second equality follows from the
chain rule for conditional mutual information. The final equality follows from
the fact that systems $B_{f}^{\prime}E$ are independent of $X_{f}$ when given
the classical systems $YZX\overline{A}$ (one can inspect the state in
\eqref{eq:big-state} to see this explicitly). Since the above chain of
inequalities holds for any non-signaling extension of the form in
\eqref{eq:particular-ext}, we can conclude that%
\begin{multline}
\sum_{z}p(z)\inf_{\text{ext.~in \eqref{eq:generic-ext-lower-bnd}}}I(X_{f}%
\overline{A}_{f};B_{f}^{\prime}|EY)_{\sigma^{z}}\\\leq\inf_{\text{ext.~in
}\eqref{eq:particular-ext}}I(X\overline{A};B_{f}^{\prime}|EYZ)_{\zeta}.
\end{multline}
Now we can take the supremum of both sides with respect to 1W-LOCC\ operations
$\{p_{X_{f}|YZ=z},\{\mathcal{L}_{y}^{(z)}\}_{y}\}_{z}$ and we find that%
\begin{multline}
\sup_{\{p_{X_{f}|YZ=z},\{\mathcal{L}_{y}^{(z)}\}_{y}\}_{z}}\sum_{z}p(z)\inf
_{\text{ext.~in \eqref{eq:generic-ext-lower-bnd}}}I(X_{f}\overline{A}_{f}%
;B_{f}^{\prime}|EY)_{\sigma^{z}}\\\leq\sup_{\{p_{X_{f}|YZ=z},\{\mathcal{L}_{y}^{(z)}%
\}_{y}\}_{z}}\inf_{\text{ext.~in }\eqref{eq:particular-ext}}I(X\overline{A}%
;B_{f}^{\prime}|EYZ)_{\zeta}. \label{eq:almost-there-1W-LOCC}%
\end{multline}
Since the 1W-LOCC\ operation $\{p_{X_{f}|YZ=z},\{\mathcal{L}_{y}^{(z)}\}_{y}\}_{z}$ is a
particular 1W-LOCC\ operation that can be performed on the original assemblage
$\{\hat{\rho}_{B}^{a,x}\}_{a,x}$, we find that%
\begin{multline}
\sup_{\{p_{X_{f}|YZ=z},\{\mathcal{L}_{y}^{(z)}\}_{y}\}_{z}}\inf_{\text{ext.~in
}\eqref{eq:particular-ext}}I(X\overline{A};B_{f}^{\prime}|EYZ)_{\zeta}\\
\leq
S(\overline{A};B)_{\hat{\rho}}.
\end{multline}
Since each $z$-dependent 1W-LOCC\ operation $\{p_{X_{f}|YZ=z},\{\mathcal{L}_{y}%
^{(z)}\}_{y}\}$ depends only on a particular value of $z$, we can then
exchange the supremum and the sum over $z$ in \eqref{eq:almost-there-1W-LOCC}
to conclude that%
\begin{align}
&  \sup_{\{p_{X_{f}|YZ=z},\{\mathcal{L}_{y}^{(z)}\}_{y}\}_{z}}\sum_{z}p(z)\inf
_{\text{ext.~in \eqref{eq:generic-ext-lower-bnd}}}I(X_{f}\overline{A}_{f}%
;B_{f}^{\prime}|EY)_{\sigma^{z}}\nonumber\\
&  =\sum_{z}p(z)\sup_{\{p_{X_{f}|YZ=z},\{\mathcal{L}_{y}^{(z)}\}_{y}\}}\inf
_{\text{ext.~in \eqref{eq:generic-ext-lower-bnd}}}I(X_{f}\overline{A}_{f}%
;B_{f}^{\prime}|EY)_{\sigma^{z}}\\
&  =\sum_{z}p(z)S(\overline{A}_{f};B_{f})_{\hat{\rho}_{z}}.
\end{align}
Putting these last steps together, we conclude \eqref{eq:1W-LOCC-monotone}.
\end{proof}

\begin{proposition}[Convexity]
\label{prop:convexity-IS}Let $\{\hat{\rho}_{B}^{a,x}\}_{a,x}$ and
$\{\hat{\sigma}_{B}^{a,x}\}_{a,x}$ be assemblages, and let $\lambda\in
\lbrack0,1]$. Let $\{\hat{\tau}_{B}^{a,x}\}_{a,x}$ be a mixture of the two
assemblages, defined as
\begin{equation}
\hat{\tau}_{B}^{a,x}:=\lambda\hat{\rho}_{B}^{a,x}+(1-\lambda)\hat{\sigma}%
_{B}^{a,x}.
\end{equation}
Then
\begin{equation}
S(\overline{A};B)_{\hat{\tau}}\leq\lambda S(\overline{A};B)_{\hat{\rho}}+(1-\lambda
)S(\overline{A};B)_{\hat{\sigma}}.
\end{equation}
\end{proposition}

\begin{proof}
We first give a proof sketch for the convexity of intrinsic steerability. Let $\lambda\in\lbrack0,1]$. Let $\{\hat{\rho}_{B}%
^{a,x}\}_{a,x}$ and $\{\hat{\sigma}_{B}^{a,x}\}_{a,x}$ be two assemblages, and
consider an assemblage $\{\hat{\tau}_{B}^{a,x}:=\lambda\hat{\rho}_{B}%
^{a,x}+(1-\lambda)\hat{\sigma}_{B}^{a,x}\}_{a,x}$. Convexity of the intrinsic
steerability is the following statement:
\begin{equation}
S(\overline{A};B)_{\hat{\tau}}\leq\lambda S(\overline{A};B)_{\hat{\rho}}+(1-\lambda
)S(\overline{A};B)_{\hat{\sigma}}, \label{convexity}%
\end{equation}
whose physical interpretation is that steering cannot increase when mixing two
assemblages. A proof for convexity is similar to known proofs for the
convexity of squashed entanglement \cite{Christandl2003} and the squashed
entanglement of a channel \cite{GEW16}. To prove convexity, first consider
arbitrary non-signaling extensions of $\{\hat{\rho}_{B}^{a,x}\}_{a,x}$ and
$\{\hat{\sigma}_{B}^{a,x}\}_{a,x}$. 
Embedding these in a larger classical--quantum state with a label chosen
according to $\lambda$ gives a particular non-signaling extension of
$\hat{\tau}$. Convexity then follows from a property of conditional mutual
information and because the intrinsic steerability involves an infimum over
all non-signaling extensions. 

We now give a detailed proof. 
Let $\{p_{X|Y},\{\mathcal{K}_y\}_{y}\}$
denote an arbitrary 1W-LOCC operation, which leads to the following
classical--quantum state:
\begin{multline}
\tau_{X\overline{A}B^{\prime}Y}:=\sum_{a,x,y}p_{X|Y}(x|y)\op{x}{x}_{X}%
\otimes\op{a}{a}_{\overline{A}}\\
\otimes \mathcal{K}_y(\hat{\tau}_{B}^{a,x})%
\otimes\op{y}{y}_{Y}.
\end{multline}
An \textit{arbitrary} non-signaling extension of this state, is as follows:
\begin{multline}
\tau_{X\overline{A}B^{\prime}YE}:=\sum_{a,x,y}p_{X|Y}(x|y)\op{x}{x}_{X}%
\otimes\op{a}{a}_{\overline{A}}
\\
\otimes\hat{\tau}_{B^{\prime}E}^{a,x,y}%
\otimes\op{y}{y}_{Y}, \label{eq:general-tau-ext}%
\end{multline}
where%
\begin{align}
\Tr_{E}(\hat{\tau}_{B^{\prime}E}^{a,x,y})  &  =\mathcal{K}_y(\hat{\tau}_{B}^{a,x}),\\
\sum_{a}\hat{\tau}_{B^{\prime}E}^{a,x,y}  &  =\sum_{a}\hat{\tau}_{B^{\prime}%
E}^{a,x^{\prime},y}\ \ \ \ \forall x,x^{\prime}\in\mathcal{X},\ y\in
\mathcal{Y}.
\end{align}
Let $\hat{\rho}_{B^{\prime}E}^{a,x,y}$ and $\hat{\sigma}_{B^{\prime}E}%
^{a,x,y}$ be \textit{arbitrary} non-signaling extensions of $\mathcal{K}_y(\hat{\rho
}_{B}^{a,x})$ and $\mathcal{K}_y(\hat{\sigma}_{B}^{a,x})$,
satisfying%
\begin{align}
\Tr_{E}(\hat{\rho}_{B^{\prime}E}^{a,x,y})  &  =\mathcal{K}_y(\hat{\rho}_{B}^{a,x}),\\
\sum_{a}\hat{\rho}_{B^{\prime}E}^{a,x,y}  &  =\sum_{a}\hat{\rho}_{B^{\prime}%
E}^{a,x^{\prime},y}\ \ \ \ \forall x,x^{\prime}\in\mathcal{X},\ y\in
\mathcal{Y},\\
\Tr_{E}(\hat{\sigma}_{B^{\prime}E}^{a,x,y})  &  =\mathcal{K}_y(\hat{\sigma}_{B}%
^{a,x}),\\
\sum_{a}\hat{\sigma}_{B^{\prime}E}^{a,x,y}  &  =\sum_{a}\hat{\sigma
}_{B^{\prime}E}^{a,x^{\prime},y}\ \ \ \ \forall x,x^{\prime}\in\mathcal{X}%
,\ y\in\mathcal{Y}.
\end{align}
These lead to the following states:%
\begin{align}
\rho_{X\overline{A}B^{\prime}YE}  &  :=\sum_{a,x,y}p_{X|Y}(x|y)\op{x}{x}_{X}%
\otimes\op{a}{a}_{\overline{A}}\nonumber\\
& \qquad \otimes\hat{\rho}_{B^{\prime}E}^{a,x,y}%
\otimes\op{y}{y}_{Y},\label{eq:gen-rho-ext}\\
\sigma_{X\overline{A}B^{\prime}YE}  &  :=\sum_{a,x,y}p_{X|Y}(x|y)\op{x}{x}_{X}%
\otimes\op{a}{a}_{\overline{A}}\nonumber \\
& \qquad \otimes\hat{\sigma}_{B^{\prime}E}^{a,x,y}%
\otimes\op{y}{y}_{Y}. \label{eq:gen-sig-ext}%
\end{align}
\begin{widetext}
A \textit{particular} non-signaling extension $\tau_{X\overline{A}B^{\prime
}YEE^{\prime}}^{\prime}$ of $\tau_{\overline{A}B^{\prime}XY}$ is given by%
\begin{equation}
\tau_{X\overline{A}B^{\prime}YEE^{\prime}}^{\prime}:=\sum_{a,x,y}p_{X|Y}%
(x|y)\op{x}{x}_{X}\otimes\op{a}{a}_{\overline{A}}\otimes\left(  \lambda\hat{\rho
}_{B^{\prime}E}^{a,x,y}\otimes\op{0}{0}_{E^{\prime}}+(1-\lambda)\hat{\sigma
}_{B^{\prime}E}^{a,x,y}\otimes\op{1}{1}_{E^{\prime}}\right)  \otimes
\op{y}{y}_{Y}. \label{eq:particular-tau-ext}%
\end{equation}
Then consider that%
\begin{align}
\inf_{\text{ext. in \eqref{eq:general-tau-ext}}}I(X\overline{A};B^{\prime
}|EY)_{\tau}  &  \leq I(X\overline{A};B^{\prime}|EYE^{\prime})_{\tau^{\prime}}\\
&  =\lambda I(X\overline{A};B^{\prime}|EY)_{\rho}+(1-\lambda)I(X\overline{A};B^{\prime
}|EY)_{\sigma}.
\end{align}
Since the inequality above holds for all general non-signaling extensions of
the form in \eqref{eq:gen-rho-ext} and \eqref{eq:gen-sig-ext}, we conclude
that%
\begin{equation}
\inf_{\text{ext. in \eqref{eq:general-tau-ext}}}I(X\overline{A};B^{\prime
}|EY)_{\tau}\leq\lambda\inf_{\text{ext. in \eqref{eq:gen-rho-ext}}}I(X\overline
{A};B^{\prime}|EY)_{\rho}+(1-\lambda)\inf_{\text{ext. in
\eqref{eq:gen-sig-ext}}}I(X\overline{A};B^{\prime}|EY)_{\sigma}.
\end{equation}
Now taking a supremum over all 1W-LOCC operations, we find that%
\begin{align}
S(\overline{A};B)_{\hat{\tau}}  &  =\sup_{\left\{  p_{X|Y},\left\{  \mathcal{K}_y\right\}
_{y}\right\}  }\inf_{\text{ext. in \eqref{eq:general-tau-ext}}}I(X\overline
{A};B^{\prime}|EY)_{\tau}\\
&  \leq\sup_{\left\{  p_{X|Y},\left\{  \mathcal{K}_y\right\}  _{y}\right\}  }\left(
\lambda\inf_{\text{ext. in \eqref{eq:gen-rho-ext}}}I(X\overline{A};B^{\prime
}|EY)_{\rho}+(1-\lambda)\inf_{\text{ext. in \eqref{eq:gen-sig-ext}}}I(X\overline
{A};B^{\prime}|EY)_{\sigma}\right) \\
&  \leq\lambda\sup_{\left\{  p_{X|Y},\left\{  \mathcal{K}_y\right\}  _{y}\right\}
}\inf_{\text{ext. in \eqref{eq:gen-rho-ext}}}I(X\overline{A};B^{\prime}|EY)_{\rho
}+(1-\lambda)\sup_{\left\{  p_{X|Y},\left\{  \mathcal{K}_y\right\}  _{y}\right\}
}\inf_{\text{ext. in \eqref{eq:gen-sig-ext}}}I(X\overline{A};B^{\prime}%
|EY)_{\sigma}\\
&  =\lambda S(\overline{A};B)_{\hat{\rho}}+(1-\lambda)S(\overline{A};B)_{\hat{\sigma}}.
\end{align}
\end{widetext}
This concludes the proof.
\end{proof}

We now consider a superadditivity property of assemblages, which holds for
intrinsic steerability. Suppose that Alice has two quantum systems $A_{1}$ and
$A_{2}$ and suppose that Bob has two quantum systems $B_{1}$ and $B_{2}$.
Alice could perform a local measurement on $A_{1}$ chosen according to $x_{1}$
and with output $a_{1}$. Similarly, Alice could perform a local measurement on
$A_{2}$ chosen according to $x_{2}$ and with output $a_{2}$. This process
realizes a joint assemblage $\{\hat{\rho}_{B_{1}B_{2}}^{a_{1},a_{2}%
,x_{1},x_{2}}\}_{a_{1},a_{2},x_{1},x_{2}}$ obeying certain no-signaling
constraints, but it also realizes some local assemblages as well. One would
expect that the steering available in the joint assemblage should never be
smaller than the sum of the steering available in the local assemblages, and
this is what the following proposition addresses:
\begin{proposition}
[Superadditivity]\label{prop:superadd}Let $\{\hat{\rho}_{B_{1}B_{2}}%
^{a_{1},a_{2},x_{1},x_{2}}\}_{a_{1},a_{2},x_{1},x_{2}}$ be an assemblage for
which the following additional no-signaling constraints hold%
\begin{align*}
\sum_{a_{2}}\hat{\rho}_{B_{1}B_{2}}^{a_{1},a_{2},x_{1},x_{2}}  &  =\sum
_{a_{2}}\hat{\rho}_{B_{1}B_{2}}^{a_{1},a_{2},x_{1},x_{2}^{\prime}}%
:=\hat{\theta}_{B_{1}B_{2}}^{a_{1},x_{1}}\ \ \ \ \forall x_{2},x_{2}^{\prime
},\\
\sum_{a_{1}}\hat{\rho}_{B_{1}B_{2}}^{a_{1},a_{2},x_{1},x_{2}}  &  =\sum
_{a_{1}}\hat{\rho}_{B_{1}B_{2}}^{a_{1},a_{2},x_{1}^{\prime},x_{2}}%
:=\hat{\kappa}_{B_{1}B_{2}}^{a_{2},x_{2}}\ \ \ \ \forall x_{1},x_{1}^{\prime},
\end{align*}
Let $\{\operatorname{Tr}_{B_{2}}(\hat{\theta}_{B_{1}B_{2}}^{a_{1},x_{1}%
})\}_{a_{1},x_{1}}$ and $\{\operatorname{Tr}_{B_{1}}(\hat{\kappa}_{B_{1}B_{2}%
}^{a_{2},x_{2}})\}_{a_{2},x_{2}}$ be reduced, local assemblages arising from
the joint assemblage $\{\hat{\rho}_{B_{1}B_{2}}^{a_{1},a_{2},x_{1},x_{2}%
}\}_{a_{1},a_{2},x_{1},x_{2}}$. Then intrinsic steerability is superadditive
in the following sense:%
\begin{equation}
S(\overline{A}_{1}\overline{A}_{2};B_{1}B_{2})_{\hat{\rho}}\geq S(\overline{A}_{1}%
;B_{1})_{\hat{\theta}}+S(\overline{A}_{2};B_{2})_{\hat{\kappa}}.
\end{equation}

\end{proposition}

\begin{proof}
The core idea behind our proof of Proposition~\ref{prop:superadd} is to
exploit the chain rule for conditional mutual information. First, pick a
1W-LOCC strategy where Alice's inputs $X_{1}$ and $X_{2}$ depend only on
measurement outcomes $Y_{1}$ and $Y_{2}$ of $B_{1}$ and $B_{2}$, respectively.
The chain rule and non-negativity of conditional mutual information imply
that
\begin{multline}
I(X_{1}X_{2}\overline{A}_{1}\overline{A}_{2};B_{1}B_{2}|EY_{1}Y_{2})_{\rho}\geq\\
I(X_{1}\overline{A}_{1};B_{1}|EY_{1}Y_{2})_{\rho}
+I(X_{2}\overline{A}_{2};B_{2}|EB_{1}Y_{1}Y_{2})_{\rho}, \label{SP1}%
\end{multline}
 where system $E$ denotes a non-signaling extension system. The idea is then to
 take $EY_{2}$ as a non-signaling extension for $X_{1}\overline{A}_{1}B_{1}Y_{1}$,
 systems $EB_{1}Y_{1}$ as a non-signaling extension for $X_{2}\overline{A}_{2}%
 B_{2}Y_{2}$, and work from there.

We now give  a detailed proof. Suppose that we apply to the assemblage $\{\hat{\rho}_{B_{1}B_{2}}%
^{a_{1},a_{2},x_{1},x_{2}}\}_{a_{1},a_{2},x_{1},x_{2}}$ a general 1W-LOCC
operation \thinspace$\{p_{X_{1}X_{2}|Y},\{\mathcal{K}_y\}_{y}\}$, resulting in the
following classical--quantum state:%
\begin{widetext}
\begin{equation}
\rho_{\overline{A}_{1}X_{1}\overline{A}_{2}X_{2}YB_{1}^{\prime}B_{2}^{\prime}}%
:=\sum_{a_{1},x_{1},a_{2},x_{2},y}p_{X_{1}X_{2}|Y}(x_{1},x_{2}|y)[a_{1}%
]\otimes\lbrack x_{1}]\otimes\lbrack a_{2}]\otimes\lbrack x_{2}]\otimes\lbrack
y]\otimes \mathcal{K}_y(\hat{\rho}_{B_{1}B_{2}}^{a_{1},x_{1},a_{2},x_{2}}).
\end{equation}

Let $\hat{\rho}_{B_{1}^{\prime}B_{2}^{\prime}E}^{a_{1},x_{1},a_{2},x_{2},y}$
be a non-signaling extension of $\mathcal{K}_y(\rho_{B_{1}B_{2}}^{a_{1},x_{1}%
,a_{2},x_{2}})$ and consider the following extension of the above
state:%
\begin{equation}
\rho_{\overline{A}_{1}X_{1}\overline{A}_{2}X_{2}YB_{1}^{\prime}B_{2}^{\prime}E}%
:=\sum_{a_{1},x_{1},a_{2},x_{2},y}p_{X_{1}X_{2}|Y}(x_{1},x_{2}|y)[a_{1}%
]\otimes\lbrack x_{1}]\otimes\lbrack a_{2}]\otimes\lbrack x_{2}]\otimes\lbrack
y]\otimes\hat{\rho}_{B_{1}^{\prime}B_{2}^{\prime}E}^{a_{1},x_{1},a_{2}%
,x_{2},y}.
\end{equation}
A particular \textquotedblleft product\textquotedblright\ 1W-LOCC operation
has the form $\{p_{X_{1}|Y_{1}}p_{X_{2}|Y_{2}},\{\mathcal{L}_{y_{1}}\otimes \mathcal{M}_{y_{2}%
}\}_{y_{1},y_{2}}\}$ and results in the following state:%
\begin{multline}
\omega_{\overline{A}_{1}X_{1}\overline{A}_{2}X_{2}Y_{1}Y_{2}B_{1}^{\prime}B_{2}^{\prime
}}:=\sum_{a_{1},x_{1},a_{2},x_{2},y}p_{X_{1}|Y_{1}}(x_{1}|y_{1})p_{X_{2}%
|Y_{2}}(x_{2}|y_{2})[a_{1}]\otimes\lbrack x_{1}]\otimes\lbrack a_{2}%
]\otimes\lbrack x_{2}]\otimes\lbrack y_{1}]\otimes\lbrack y_{2}]\\
\otimes\left(  \mathcal{L}_{y_{1}}\otimes \mathcal{M}_{y_{2}}\right)  (\hat{\rho}_{B_{1}B_{2}%
}^{a_{1},x_{1},a_{2},x_{2}}).
\end{multline}
Let $\hat{\omega}_{B_{1}^{\prime}B_{2}^{\prime}E}^{a_{1},x_{1},a_{2}%
,x_{2},y_{1},y_{2}}$ be a non-signaling extension of $\left(  \mathcal{L}_{y_{1}}\otimes
\mathcal{M}_{y_{2}}\right) ( \hat{\rho}_{B_{1}B_{2}}^{a_{1},x_{1},a_{2},x_{2}})$, and define the following state:
\begin{multline}
\omega_{\overline{A}_{1}X_{1}\overline{A}_{2}X_{2}Y_{1}Y_{2}B_{1}^{\prime}B_{2}^{\prime
}E}:=\sum_{a_{1},x_{1},a_{2},x_{2},y}p_{X_{1}|Y_{1}}(x_{1}|y_{1}%
)p_{X_{2}|Y_{2}}(x_{2}|y_{2})[a_{1}]\otimes\lbrack x_{1}]\otimes\lbrack
a_{2}]\otimes\lbrack x_{2}]\otimes\lbrack y_{1}]\otimes\lbrack y_{2}%
]\label{eq:arbitrary-omega-super}\\
\otimes\hat{\omega}_{B_{1}^{\prime}B_{2}^{\prime}E}^{a_{1},x_{1},a_{2}%
,x_{2},y_{1},y_{2}}.
\end{multline}
Let $\hat{\theta}_{B_{1}^{\prime}F}^{a_{1},x_{1},y_{1}}$ be a non-signaling
extension of $\mathcal{L}_{y_{1}}(\hat{\theta}_{B_{1}}^{a_{1},x_{1}})$ and
let $\hat{\kappa}_{B_{2}^{\prime}G}^{a_{2},x_{2},y_{2}}$ be a non-signaling
extension of $\mathcal{M}_{y_{2}}(\hat{\kappa}_{B_{2}}^{a_{2},x_{2}})$,
leading to the following classical--quantum states:%
\begin{align}
\theta_{X_{1}\overline{A}_{1}B_{1}^{\prime}FY_{1}}  &  :=\sum_{x_{1},a_{1}}%
p_{X_{1}|Y_{1}}(x_{1}|y_{1})[x_{1}]\otimes\lbrack a_{1}]\otimes\hat{\theta
}_{B_{1}^{\prime}F}^{a_{1},x_{1},y_{1}}\otimes\lbrack y_{1}%
],\label{eq:theta-ext-super}\\
\kappa_{X_{2}\overline{A}_{2}B_{2}^{\prime}GY_{2}}  &  :=\sum_{x_{2},a_{2}}%
p_{X_{2}|Y_{2}}(x_{2}|y_{2})[x_{2}]\otimes\lbrack a_{2}]\otimes\hat{\kappa
}_{B_{2}^{\prime}G}^{a_{2},x_{2},y_{2}}\otimes\lbrack y_{2}].
\label{eq:kappa-ext-super}%
\end{align}
\end{widetext}
Consider that%
\begin{align}
&  I(\overline{A}_{1}X_{1}\overline{A}_{2}X_{2};B_{1}^{\prime}B_{2}^{\prime}|EY_{1}%
Y_{2})_{\omega}\nonumber\\
&  =I(\overline{A}_{1}X_{1}\overline{A}_{2}X_{2};B_{1}^{\prime}|EY_{1}Y_{2})_{\omega
}\nonumber
\\&\quad+I(\overline{A}_{1}X_{1}\overline{A}_{2}X_{2};B_{2}^{\prime}|EB_{1}^{\prime}Y_{1}%
Y_{2})_{\omega}\\
&  =I(\overline{A}_{1}X_{1};B_{1}^{\prime}|EY_{1}Y_{2})_{\omega}+I(\overline{A}_{2}%
X_{2};B_{1}^{\prime}|EY_{1}Y_{2}\overline{A}_{1}X_{1})_{\omega}\nonumber\\
& \quad+(\overline{A}_{2}X_{2};B_{2}^{\prime}|EB_{1}^{\prime}Y_{1}Y_{2}%
)_{\omega}\nonumber\\ &\quad+I(\overline{A}_{1}X_{1};B_{2}^{\prime}|EB_{1}^{\prime}Y_{1}Y_{2}\overline
{A}_{2}X_{2})_{\omega}\\
&  \geq I(\overline{A}_{1}X_{1};B_{1}^{\prime}|EY_{1}Y_{2})_{\omega}+I(\overline{A}%
_{2}X_{2};B_{2}^{\prime}|EB_{1}^{\prime}Y_{1}Y_{2})_{\omega}\\
&  \geq\inf_{\text{ext. in \eqref{eq:theta-ext-super}}}I(\overline{A}_{1}%
X_{1};B_{1}^{\prime}|FY_{1})_{\theta}\nonumber\\&\quad+\inf_{\text{ext. in
\eqref{eq:kappa-ext-super}}}I(\overline{A}_{2}X_{2};B_{2}^{\prime}|GY_{2})_{\kappa
}.
\end{align}
The first two equalities follow from the chain rule for conditional mutual
information. The first inequality follows by dropping two of the terms and
from the fact that the conditional mutual information is non-negative. To see
the last inequality, consider that the state $\sum_{a_{2},x_{2},y_{2}}%
\hat{\omega}_{B_{1}^{\prime}E}^{a_{1},x_{1},a_{2},x_{2},y_{1},y_{2}}%
\otimes\lbrack y_{2}]$ is a particular non-signaling extension of $\mathcal{L}_{y_{1}%
}(\hat{\theta}_{B_{1}}^{a_{1},x_{1}})$ and the state
$\sum_{a_{1},x_{1},y_{1}}\hat{\omega}_{B_{1}^{\prime}B_{2}^{\prime}E}%
^{a_{1},x_{1},a_{2},x_{2},y_{1},y_{2}}\otimes\lbrack y_{1}]$ is a particular
non-signaling extension of $\mathcal{M}_{y_{2}}(\hat{\kappa}_{B_{2}}^{a_{2},x_{2}%
})$, such that an infimization over arbitrary respective
non-signaling extensions $\hat{\theta}_{B_{1}^{\prime}F}^{a_{1},x_{1},y_{1}}$
and $\hat{\kappa}_{B_{2}^{\prime}G}^{a_{2},x_{2},y_{2}}$ can never lead to
higher values of the conditional mutual informations. Since we have shown the
inequality above for an arbitrary non-signaling extension $\hat{\omega}%
_{B_{1}^{\prime}B_{2}^{\prime}E}^{a_{1},x_{1},a_{2},x_{2},y_{1},y_{2}}$, we
can conclude that%
\begin{multline}
\inf_{\text{ext. in \eqref{eq:arbitrary-omega-super}}}I(\overline{A}_{1}X_{1}%
\overline{A}_{2}X_{2};B_{1}^{\prime}B_{2}^{\prime}|EY_{1}Y_{2})_{\omega}\\
\geq
\inf_{\text{ext. in \eqref{eq:theta-ext-super}}}I(\overline{A}_{1}X_{1}%
;B_{1}^{\prime}|FY_{1})_{\theta}\\
+\inf_{\text{ext. in
\eqref{eq:kappa-ext-super}}}I(\overline{A}_{2}X_{2};B_{2}^{\prime}|GY_{2})_{\kappa
},
\end{multline}
which in turn implies that%
\begin{widetext}
\begin{multline}
\sup_{\{p_{X_{1}|Y_{1}}p_{X_{2}|Y_{2}},\{\mathcal{L}_{y_{1}}\otimes \mathcal{M}_{y_{2}}%
\}_{y_{1},y_{2}}\}}\inf_{\text{ext. in \eqref{eq:arbitrary-omega-super}}%
}I(\overline{A}_{1}X_{1}\overline{A}_{2}X_{2};B_{1}^{\prime}B_{2}^{\prime}|EY_{1}%
Y_{2})_{\omega}\\
\geq\inf_{\text{ext. in \eqref{eq:theta-ext-super}}}I(\overline{A}_{1}X_{1}%
;B_{1}^{\prime}|FY_{1})_{\theta}+\inf_{\text{ext. in
\eqref{eq:kappa-ext-super}}}I(\overline{A}_{2}X_{2};B_{2}^{\prime}|GY_{2})_{\kappa
}.
\end{multline}
The reduced 1W-LOCC\ operations $\{p_{X_{1}|Y_{1}},\{\mathcal{L}_{y_{1}}\}_{y_{1}}\}$
and $\{p_{X_{2}|Y_{2}},\{\mathcal{M}_{y_{2}}\}_{y_{2}}\}$ are arbitrary, and so we can
conclude that%
\begin{align}
&  \!\!\!\!\!\!\!\!\!\!\sup_{\{p_{X_{1}|Y_{1}}p_{X_{2}|Y_{2}},\{\mathcal{L}_{y_{1}}\otimes \mathcal{M}_{y_{2}%
}\}_{y_{1},y_{2}}\}}\inf_{\text{ext. in \eqref{eq:arbitrary-omega-super}}%
}I(\overline{A}_{1}X_{1}\overline{A}_{2}X_{2};B_{1}^{\prime}B_{2}^{\prime}|EY_{1}%
Y_{2})_{\omega}\nonumber\\
&  \geq\sup_{\{p_{X_{1}|Y_{1}},\{\mathcal{L}_{y_{1}}\}_{y_{1}}\}}\inf_{\text{ext. in
\eqref{eq:theta-ext-super}}}I(\overline{A}_{1}X_{1};B_{1}^{\prime}|FY_{1})_{\theta
}+\sup_{\{p_{X_{2}|Y_{2}},\{\mathcal{M}_{y_{2}}\}_{y_{2}}\}}\inf_{\text{ext. in
\eqref{eq:kappa-ext-super}}}I(\overline{A}_{2}X_{2};B_{2}^{\prime}|GY_{2})_{\kappa
}\\
&  =S(\overline{A}_{1};B_{1})_{\hat{\theta}}+S(\overline{A}_{2};B_{2})_{\hat{\kappa}}.
\end{align}
\end{widetext}
Finally, since the 1W-LOCC\ operation $\{p_{X_{1}|Y_{1}}p_{X_{2}|Y_{2}%
},\{\mathcal{L}_{y_{1}}\otimes \mathcal{M}_{y_{2}}\}_{y_{1},y_{2}}\}$ has a particular product
form, we could never achieve a lower value of the quantity on the LHS\ by
allowing for an arbitrary 1W-LOCC\ operation, implying the desired
superadditivity:%
\begin{equation}
S(\overline{A}_{1}\overline{A}_{2};B_{1}B_{2})_{\hat{\rho}}\geq S(\overline{A}_{1}%
;B_{1})_{\hat{\theta}}+S(\overline{A}_{2};B_{2})_{\hat{\kappa}}.
\end{equation}
This concludes the proof.
\end{proof}

\section{Restricted Instrinsic Steerability}

\label{sec:rest}

As stated above, we also consider
a steering quantifier relevant in the context of restricted 1W-LOCC. Here we give a proof of Theorem \ref{thm:RIS-monotone} and proofs of various other properties of restricted intrinsic steerability.

 \begin{proposition}
 The restricted intrinsic steerability vanishes for
 an assemblage having a local-hidden state model.
 \end{proposition}
 \begin{proof}
 To prove this, consider the following non-signaling, classical extension of an
unsteerable assemblage:
\begin{multline}
\rho_{X\overline{A}BE}:=\sum_{a,x}p_{X}(x)\op{x}{ x}_{X}\otimes p_{\overline{A}%
|X\Lambda}(a|x,\lambda)\op{ a}{ a}_{\overline{A}}\\
\otimes\hat{\rho}_{B}^{\lambda}\otimes p_{\Lambda}(\lambda
)\op{\lambda}{\lambda}_{E}.\label{eq:classical-extension}%
\end{multline}
 Then 
$I(X\overline{A};B|E)_{\rho}=\sum_{\lambda}p_{\Lambda}(\lambda)I(X\overline
{A};B)_{\rho^{\lambda}}$, where 
\begin{equation}
\rho_{X\overline{A}B}^{\lambda}=\sum_{a,x}p_{X}(x)\op{ x}{ x}_{X}\otimes p_{\overline
{A}|X\Lambda}(a|x,\lambda)\op{ a}{ a}_{\overline{A}}\otimes\rho_{B}^{\lambda
},
\end{equation}
and we have used the fact that the conditional mutual information can be
written as a convex combination of mutual informations for a classical
conditioning system. By inspection, we see that systems $X\overline{A}$ and $B$ are
independent when given the shared variable $\Lambda=\lambda$. By choosing
system $E$ to contain the shared random variable $\Lambda$, the result is that
the systems form a Markov chain $X\overline{A}-E-B$, so that the conditional mutual
information $I(X\overline{A};B|E)_{\rho}$ is equal to zero. Since this argument
holds for any probability distribution $p_{X}$, we conclude that $S^{R}%
(\overline{A};B)_{\hat{\rho}}=0$.
\end{proof}

\begin{proposition}
[Restricted 1W-LOCC monotone]Let $\{\hat{\rho}_{B}^{a,x}\}_{a,x}$ be an
assemblage, and let 
\begin{equation}
\{p_{X|X_{f}},p_{\overline{A}_{f}|\overline{A}XX_{f}Z}%
,\{\mathcal{K}_{z}\}_{z}\}
\end{equation}
 denote a restricted 1W-LOCC\ operation that results in an
assemblage $\{\hat{\sigma}_{B^{\prime}}^{a_{f},x_{f}}\}_{a_{f},x_{f}}$,
defined as%
\begin{multline}
\hat{\sigma}_{B^{\prime}}^{a_{f},x_{f}}:=\\
\sum_{a,x,z}p_{X|X_{f}}%
(x|x_{f})p_{\overline{A}_{f}|\overline{A}XX_{f}Z}(a_{f}|a,x,x_{f},z)\mathcal{K}_z(\hat{\rho}%
_{B}^{a,x}).
\end{multline}
Then%
\begin{equation}
S^{R}(\overline{A};B)_{\hat{\rho}}\geq S^{R}(\overline{A}_{f};B^{\prime})_{\hat{\sigma}%
}.
\end{equation}
\end{proposition}

\begin{proof}
Taking a distribution $p_{X_{f}}$ over the black-box inputs of the final
assemblage, we can embed the state of the final assemblage into the following
classical--quantum state:%
\begin{equation}
\sigma_{X_{f}\overline{A}_{f}B^{\prime}}:=\sum_{x_{f},a_{f}}p_{X_{f}}(x_{f}%
)[x_{f}]\otimes\lbrack a_{f}]\otimes\hat{\sigma}_{B^{\prime}}^{a_{f},x_{f}},
\label{eq:cq-final-state-RIS}%
\end{equation}
which is a marginal of the following state:%
\begin{multline}
\sigma_{X_{f}X\overline{A}_{f}\overline{A}ZB^{\prime}}:=\sum_{x_{f},a_{f},a,x,z}%
p_{X_{f}}(x_{f})[x_{f}]\\
\otimes p_{X|X_{f}}(x|x_{f})[x]\otimes p_{\overline{A}%
_{f}|\overline{A}XX_{f}Z}(a_{f}|a,x,x_{f},z)[a_{f}]
\\\otimes\lbrack a]\otimes\lbrack
z]\otimes \mathcal{K}_z(\hat{\rho}_{B}^{a,x}).
\end{multline}
An \textit{arbitrary} non-signaling extension of the state in
\eqref{eq:cq-final-state-RIS} 
is as follows:%
\begin{equation}
\sigma_{X_{f}\overline{A}_{f}B^{\prime}E}:=\sum_{x_{f},a_{f}}p_{X_{f}}(x_{f}%
)[x_{f}]\otimes\lbrack a_{f}]\otimes\hat{\sigma}_{B^{\prime}E}^{a_{f},x_{f}},
\label{eq:gen-sig-ext-restr}%
\end{equation}
where%
\begin{align}
\operatorname{Tr}_{E}(\hat{\sigma}_{B^{\prime}E}^{a_{f},x_{f}})  &
=\hat{\sigma}_{B^{\prime}}^{a_{f},x_{f}},\\
\sum_{a_{f}}\hat{\sigma}_{B^{\prime}E}^{a_{f},x_{f}}  &  =\sum_{a_{f}}%
\hat{\sigma}_{B^{\prime}E}^{a_{f},x_{f}}\ \ \ \ \forall x_{f},x_{f}^{\prime
}\in\mathcal{X}_{f}.
\end{align}
\begin{widetext}
A \textit{particular} non-signaling extension of the state in
\eqref{eq:cq-final-state-RIS} 
is as follows:%
\begin{equation}
\omega_{X_{f}\overline{A}_{f}B^{\prime}EZ}:=\sum_{x_{f},a_{f}}p_{X_{f}}%
(x_{f})[x_{f}]\otimes\lbrack a_{f}]\otimes\sum_{x_{f},a_{f},a,x,z}p_{X|X_{f}%
}(x|x_{f})p_{\overline{A}_{f}|\overline{A}XX_{f}Z}(a_{f}|a,x,x_{f},z)\mathcal{K}_z(\hat{\rho
}_{BE}^{a,x})\otimes\lbrack z], \label{eq:RIS-mono-part-ext}%
\end{equation}
where%
\begin{equation}
\operatorname{Tr}_{E}(\hat{\rho}_{BE}^{a,x})    =\hat{\rho}_{B}^{a,x},\qquad \qquad 
\sum_{a}\hat{\rho}_{BE}^{a,x}    =\sum_{a}\hat{\rho}_{BE}^{a,x^{\prime}%
}\ \ \ \ \forall x,x^{\prime}\in\mathcal{X}.
\end{equation}
The state $\omega_{X_{f}\overline{A}_{f}B^{\prime}E}$ is a marginal of the
following state:%
\begin{equation}
\omega_{X_{f}X\overline{A}_{f}\overline{A}B^{\prime}EZ}:=\sum_{x_{f},a_{f}%
,a,x,z}p_{X_{f}}(x_{f})[x_{f}]\otimes p_{X|X_{f}}(x|x_{f})[x]\otimes
p_{\overline{A}_{f}|\overline{A}XX_{f}Z}(a_{f}|a,x,x_{f},z)[a_{f}]\otimes\lbrack
a]\otimes \mathcal{K}_z(\hat{\rho}_{BE}^{a,x})\otimes\lbrack z].
\label{eq:big-omega-state-restr}%
\end{equation}
\end{widetext}
Let $\rho_{X\overline{A}BE}$ be the following state:%
\begin{multline}
\rho_{X\overline{A}BE}:=\sum_{x_{f},a,x}p_{X_{f}}(x_{f})[x_{f}]\otimes p_{X|X_{f}%
}(x|x_{f})[x]\\
\otimes\lbrack a]\otimes\hat{\rho}_{BE}^{a,x}.
\label{eq:init-state-ext-restr}%
\end{multline}
Consider that%
\begin{align}
&\inf_{\text{ext. in \eqref{eq:gen-sig-ext-restr}}}I(X_{f}\overline{A}_{f}%
;B^{\prime}|E)_{\sigma}  \nonumber \\
&  \leq I(X_{f}\overline{A}_{f};B^{\prime}|EZ)_{\omega}\\
&  \leq I(X_{f}\overline{A}_{f}X\overline{A};B^{\prime}|EZ)_{\omega}\\
&  =I(X\overline{A};B^{\prime}|EZ)_{\omega}+I(X_{f};B^{\prime}|EZX\overline{A})_{\omega
}\nonumber \\
& \qquad +I(\overline{A}_{f};B^{\prime}|EZX_{f}X\overline{A})_{\omega}\\
&  =I(X\overline{A};B^{\prime}|EZ)_{\omega}\\
&  \leq I(X\overline{A};B^{\prime}Z|E)_{\omega}\\
&  \leq I(X\overline{A};B|E)_{\rho}.
\end{align}
The first inequality follows because the non-signaling extension in
\eqref{eq:RIS-mono-part-ext} is a particular kind of non-signaling extension.
The second inequality follows from data processing. The first equality follows
from the chain rule for conditional mutual information. The second equality
follows from various Markov-chain structures when inspecting
\eqref{eq:big-omega-state-restr}:\ $X_{f}$ is independent of $B^{\prime}E$
when given $ZX\overline{A}$, and $\overline{A}_{f}$ is independent of $B^{\prime}E$ when
given $ZX_{f}X\overline{A}$, so that $I(X_{f};B^{\prime}|EZX\overline{A})_{\omega
}=I(\overline{A}_{f};B^{\prime}|EZX_{f}X\overline{A})_{\omega}=0$. The third inequality
follows by applying the chain rule for and non-negativity of conditional
mutual information. The last inequality follows again from data processing.
Since the inequality holds for all non-signaling extensions of the form in
\eqref{eq:init-state-ext-restr}, we can conclude that%
\begin{align}
& \inf_{\text{ext. in \eqref{eq:gen-sig-ext-restr}}}I(X_{f}\overline{A}_{f}%
;B^{\prime}|E)_{\sigma}  \nonumber \\
  & \leq\inf_{\text{ext. in }%
\eqref{eq:init-state-ext-restr}}I(X\overline{A};B|E)_{\rho}\\
&  \leq\sup_{p_{X}}\inf_{\text{ext. in }\eqref{eq:init-state-ext-restr}}%
I(X\overline{A};B|E)_{\rho}.
\end{align}
Since the inequality above holds for an arbitrary choice of $p_{X_{f}}$, we
can finally conclude that%
\begin{multline}
\sup_{p_{X_{f}}}\inf_{\text{ext. in \eqref{eq:gen-sig-ext-restr}}}I(X_{f}%
\overline{A}_{f};B^{\prime}|E)_{\sigma}\\
\leq\sup_{p_{X}}\inf_{\text{ext. in
}\eqref{eq:init-state-ext-restr}}I(X\overline{A};B|E)_{\rho},
\end{multline}
which is equivalent to the statement of the proposition.
\end{proof}

\bigskip

The proof of convexity of the restricted intrinsic steerability is along the
same lines as that for intrinsic steerability, given already in the proof of
Proposition~\ref{prop:convexity-IS}. We summarize the result as the following proposition:

\begin{proposition}[Convexity]
Let $\{\hat{\rho}_{B}^{a,x}\}_{a,x}$ and $\{\hat{\sigma}_{B}^{a,x}\}_{a,x}$ be
assemblages, and let $\lambda\in[0,1]$. Let $\{\hat{\tau}_{B}^{a,x}\}_{a,x}$
be a mixture of the two assemblages, defined as
\begin{equation}
\hat{\tau}_{B}^{a,x}:=\lambda\hat{\rho}_{B}^{a,x}+(1-\lambda)\hat{\sigma}%
_{B}^{a,x}.
\end{equation}
Then
\begin{equation}
S^{R}(\overline{A};B)_{\hat{\tau}}\leq\lambda S^{R}(\overline{A};B)_{\hat{\rho}%
}+(1-\lambda)S^{R}(\overline{A};B)_{\hat{\sigma}}.
\end{equation}

\end{proposition}

\begin{proposition}[Superadditivity and Additivity]
\label{superadditivityr}
Let $\{\hat{\rho}_{B_{1}B_{2}}^{a_{1},a_{2}%
,x_{1},x_{2}}\}_{a_{1},a_{2},x_{1},x_{2}}$ be an assemblage for which the
following additional no-signaling constraints hold%
\begin{align}
\sum_{a_{2}}\hat{\rho}_{B_{1}B_{2}}^{a_{1},a_{2},x_{1},x_{2}}  &  =\sum
_{a_{2}}\hat{\rho}_{B_{1}B_{2}}^{a_{1},a_{2},x_{1},x_{2}^{\prime}}%
:=\hat{\theta}_{B_{1}B_{2}}^{a_{1},x_{1}}\ \ \ \ \forall x_{2},x_{2}^{\prime
},\\
\sum_{a_{1}}\hat{\rho}_{B_{1}B_{2}}^{a_{1},a_{2},x_{1},x_{2}}  &  =\sum
_{a_{1}}\hat{\rho}_{B_{1}B_{2}}^{a_{1},a_{2},x_{1}^{\prime},x_{2}}%
:=\hat{\kappa}_{B_{1}B_{2}}^{a_{2},x_{2}}\ \ \ \ \forall x_{1},x_{1}^{\prime},
\end{align}
Let $\{\operatorname{Tr}_{B_{2}}(\hat{\theta}_{B_{1}B_{2}}^{a_{1},x_{1}%
})\}_{a_{1},x_{1}}$ and $\{\operatorname{Tr}_{B_{1}}(\hat{\kappa}_{B_{1}B_{2}%
}^{a_{2},x_{2}})\}_{a_{2},x_{2}}$ be reduced assemblages arising from the
joint assemblage $\{\hat{\rho}_{B_{1}B_{2}}^{a_{1},a_{2},x_{1},x_{2}}%
\}_{a_{1},a_{2},x_{1},x_{2}}$. Then the restricted intrinsic steerability is
superadditive in the following sense:%
\begin{equation}
S^{R}(\overline{A}_{1}\overline{A}_{2};B_{1}B_{2})_{\hat{\rho}}\geq S^{R}(\overline{A}%
_{1};B_{1})_{\hat{\theta}}+S^{R}(\overline{A}_{2};B_{2})_{\hat{\kappa}}.
\end{equation}
If the assemblage $\{\hat{\rho}_{B_{1}B_{2}}^{a_{1},a_{2},x_{1},x_{2}%
}\}_{a_{1},a_{2},x_{1},x_{2}}$ has a tensor-product form, so that $\hat{\rho
}_{B_{1}B_{2}}^{a_{1},a_{2},x_{1},x_{2}}=\hat{\theta}_{B_{1}}^{a_{1},x_{1}%
}\otimes\hat{\kappa}_{B_{2}}^{a_{2},x_{2}}$ for assemblages $\{\hat{\theta
}_{B_{1}}^{a_{1},x_{1}}\}_{a_{1},x_{1}}$ and $\{\hat{\kappa}_{B_{2}}%
^{a_{2},x_{2}}\}_{a_{2},x_{2}}$, then the restricted intrinsic steerability is
additive:%
\begin{equation}
S^{R}(\overline{A}_{1}\overline{A}_{2};B_{1}B_{2})_{\hat{\rho}}=S^{R}(\overline{A}_{1}%
;B_{1})_{\hat{\theta}}+S^{R}(\overline{A}_{2};B_{2})_{\hat{\kappa}}.
\end{equation}

\end{proposition}

\begin{proof}
The superadditivity of restricted intrinsic steerability is similar to the
proof of Proposition~\ref{prop:superadd} for intrinsic steerability. Thus, to prove the additivity of
intrinsic steerability with respect to product assemblages, it is sufficient
to prove the following subadditivity inequality:%
\begin{equation}
S^{R}(\overline{A}_{1}\overline{A}_{2};B_{1}B_{2})_{\hat{\rho}}\leq S^{R}(\overline{A}%
_{1};B_{1})_{\hat{\theta}}+S^{R}(\overline{A}_{2};B_{2})_{\hat{\kappa}}.
\label{eq:sub-add-ineq-restr}%
\end{equation}
Our proof of the above inequality has some similarities to the proof of the
additivity of the squashed entanglement of a channel \cite{TGW14IEEE} (there
are, however, some key differences). Let $\hat{\theta}_{B_{1}E_{1}}%
^{a_{1},x_{1}}$ and $\hat{\kappa}_{B_{2}E_{2}}^{a_{2},x_{2}}$ be non-signaling
extensions of $\hat{\theta}_{B_{1}}^{a_{1},x_{1}}$ and $\hat{\kappa}_{B_{2}%
}^{a_{2},x_{2}}$, respectively, and suppose that $|\hat{\theta}^{a_{1},x_{1}%
}\rangle_{B_{1}E_{1}F_{1}}$ and $|\hat{\kappa}^{a_{2},x_{2}}\rangle
_{B_{2}E_{2}F_{2}}$ purify $\hat{\theta}_{B_{1}E_{1}}^{a_{1},x_{1}}$ and
$\hat{\kappa}_{B_{2}E_{2}}^{a_{2},x_{2}}$, respectively. Consider the
following states:%
\begin{widetext}
\begin{align}
\rho_{X_{1}X_{2}\overline{A}_{1}\overline{A}_{2}B_{1}B_{2}E}  &  :=\sum_{x_{1}%
,x_{2},a_{1},a_{2}}p_{X_{1}X_{2}}(x_{1},x_{2})[x_{1}]\otimes\lbrack
x_{2}]\otimes\lbrack a_{1}]\otimes\lbrack a_{2}]\otimes\hat{\rho}_{B_{1}%
B_{2}E}^{a_{1},a_{2},x_{1},x_{2}},\\
\omega_{X_{1}X_{2}\overline{A}_{1}\overline{A}_{2}B_{1}B_{2}E_{1}E_{2}F_{1}F_{2}}  &
:=\sum_{x_{1},x_{2},a_{1},a_{2}}p_{X_{1}X_{2}}(x_{1},x_{2})[x_{1}%
]\otimes\lbrack x_{2}]\otimes\lbrack a_{1}]\otimes\lbrack a_{2}]\otimes
\hat{\theta}_{B_{1}E_{1}F_{1}}^{a_{1},x_{1}}\otimes\hat{\kappa}_{B_{2}%
E_{2}F_{2}}^{a_{2},x_{2}}, \label{eq:add-ind-ext-omega}%
\end{align}
where $p_{X_{1}X_{2}}(x_{1},x_{2})$ is some probability distribution and
$\operatorname{Tr}_{E}(\hat{\rho}_{B_{1}B_{2}E}^{a_{1},a_{2},x_{1},x_{2}%
})=\hat{\theta}_{B_{1}}^{a_{1},x_{1}}\otimes\hat{\kappa}_{B_{2}}^{a_{2},x_{2}%
}$. Consider that%
\begin{align}
&  \!\!\!\!\!\!\inf_{\rho_{\overline{A}_{1}\overline{A}_{2}X_{1}X_{2}B_{1}B_{2}E}}%
I(\overline{A}_{1}\overline{A}_{2}X_{1}X_{2};B_{1}B_{2}|E)_{\rho}\nonumber\\
&  \leq I(\overline{A}_{1}\overline{A}_{2}X_{1}X_{2};B_{1}B_{2}|E_{1}E_{2})_{\omega}\\
&  =H(B_{1}B_{2}|E_{1}E_{2})_{\omega}-H(B_{1}B_{2}|E_{1}E_{2}\overline{A}_{1}%
X_{1}\overline{A}_{2}X_{2})_{\omega}\\
&  =H(B_{1}B_{2}|E_{1}E_{2})_{\omega}+H(B_{1}B_{2}|F_{1}F_{2}\overline{A}_{1}%
X_{1}\overline{A}_{2}X_{2})_{\omega}\\
&  \leq H(B_{1}|E_{1})_{\omega}+H(B_{2}|E_{2})_{\omega}+H(B_{1}|F_{1}\overline
{A}_{1}X_{1})_{\omega}+H(B_{2}|F_{2}\overline{A}_{2}X_{2})_{\omega}\\
&  =H(B_{1}|E_{1})_{\omega}+H(B_{2}|E_{2})_{\omega}-H(B_{1}|E_{1}\overline{A}%
_{1}X_{1})_{\omega}-H(B_{2}|E_{2}\overline{A}_{2}X_{2})_{\omega}\\
&  =I(X_{1}\overline{A}_{1};B_{1}|E_{1})_{\omega}+I(X_{2}\overline{A}_{2};B_{2}%
|E_{2})_{\omega}.
\end{align}
The first inequality follows because $\omega_{X_{1}X_{2}\overline{A}_{1}\overline{A}%
_{2}B_{1}B_{2}E_{1}E_{2}}$ is a particular non-signaling extension whereas
$\rho_{X_{1}X_{2}\overline{A}_{1}\overline{A}_{2}B_{1}B_{2}E}$ is an arbitrary
non-signaling extension. The first equality follows from the chain rule for
conditional mutual information. Conditioned on $\overline{A}_{1}\overline{A}_{2}%
X_{1}X_{2}$, the state on $B_{1}E_{1}B_{2}E_{2}F_{1}F_{2}$ is pure, and so the
second equality follows from the duality of conditional entropy. The first
inequality is a consequence of the strong subadditivity of quantum entropy
\cite{LR73}. The third equality follows again from the duality of conditional
entropy as well as the no-signaling condition. To see this for the entropy
$H(B_{1}|F_{1}\overline{A}_{1}X_{1})_{\omega}$, consider that this entropy is
evaluated with respect to the following reduced state:%
\begin{align}
&  \operatorname{Tr}_{X_{2}\overline{A}_{2}B_{2}E_{2}F_{2}}\!\left(  \sum
_{x_{1},x_{2},a_{1},a_{2}}p_{X_{1}X_{2}}(x_{1},x_{2})[x_{1}]\otimes\lbrack
x_{2}]\otimes\lbrack a_{1}]\otimes\lbrack a_{2}]\otimes\hat{\theta}%
_{B_{1}E_{1}F_{1}}^{a_{1},x_{1}}\otimes\hat{\kappa}_{B_{2}E_{2}F_{2}}%
^{a_{2},x_{2}}\right) \nonumber\\
&  =\sum_{x_{1},x_{2},a_{1},a_{2}}p_{X_{1}X_{2}}(x_{1},x_{2})[x_{1}%
]\otimes\lbrack a_{1}]\otimes\hat{\theta}_{B_{1}E_{1}F_{1}}^{a_{1},x_{1}%
}\otimes\operatorname{Tr}_{B_{2}E_{2}F_{2}}\{\hat{\kappa}_{B_{2}E_{2}F_{2}%
}^{a_{2},x_{2}}\}\\
&  =\sum_{x_{1},a_{1}}p_{X_{1}}(x_{1})[x_{1}]\otimes\lbrack a_{1}]\otimes
\hat{\theta}_{B_{1}E_{1}F_{1}}^{a_{1},x_{1}}\otimes\operatorname{Tr}_{B_{2}%
}\!\left(  \sum_{x_{2}}p_{X_{2}|X_{1}}(x_{2}|x_{1})\sum_{a_{2}}\hat{\kappa
}_{B_{2}}^{a_{2},x_{2}}\right) \\
&  =\sum_{x_{1},a_{1}}p_{X_{1}}(x_{1})[x_{1}]\otimes\lbrack a_{1}]\otimes
\hat{\theta}_{B_{1}E_{1}F_{1}}^{a_{1},x_{1}}\otimes\operatorname{Tr}_{B_{2}%
}\!\left(  \sum_{x_{2}}p_{X_{2}|X_{1}}(x_{2}|x_{1})\kappa_{B_{2}}\right) \\
&  =\sum_{x_{1},a_{1}}p_{X_{1}}(x_{1})[x_{1}]\otimes\lbrack a_{1}]\otimes
\hat{\theta}_{B_{1}E_{1}F_{1}}^{a_{1},x_{1}}\otimes\operatorname{Tr}_{B_{2}%
}(\kappa_{B_{2}})\\
&  =\sum_{x_{1},a_{1}}p_{X_{1}}(x_{1})[x_{1}]\otimes\lbrack a_{1}]\otimes
\hat{\theta}_{B_{1}E_{1}F_{1}}^{a_{1},x_{1}}.
\end{align}
\end{widetext}
In the above, the third equality is the critical one in which we have used the
no-signaling constraint for the assemblage $\{\hat{\kappa}_{B_{2}}%
^{a_{2},x_{2}}\}_{a_{2},x_{2}}$, allowing for the effective removal of
correlation between $X_{1}$ and $X_{2}$. Thus, the above analysis allows for
seeing that the remaining state on $B_{1}E_{1}F_{1}$ conditioned on $\overline
{A}_{1}$ and $X_{1}$ is independent of any of the second system. For the last
equality, we employ the definition of conditional mutual information. Since
the above development holds for all non-signaling extensions of the form in
\eqref{eq:add-ind-ext-omega}, we find that%
\begin{align}
&\!\!\!\!\!\!\!\!\!\inf_{\rho_{\overline{A}_{1}\overline{A}_{2}X_{1}X_{2}B_{1}B_{2}E}}I(\overline{A}_{1}\overline
{A}_{2}X_{1}X_{2};B_{1}B_{2}|E_{1})_{\rho}
\nonumber \\
& \leq\inf_{\omega_{\overline{A}%
_{1}X_{1}B_{1}E_{1}}}I(\overline{A}_{1}X_{1};B_{1}|E_{1})_{\omega}\nonumber\\
&
\qquad+\inf
_{\omega_{\overline{A}_{2}X_{2}B_{2}E_{2}}}I(\overline{A}_{2}X_{2};B_{2}|E_{2})_{\omega
}\\
& \leq\sup_{p_{X_{1}}}\inf_{\omega_{\overline{A}_{1}X_{1}B_{1}E_{1}}}I(\overline{A}%
_{1}X_{1};B_{1}|E_{1})_{\omega}\nonumber\\
&
\qquad+\sup_{p_{X_{2}}}\inf_{\omega_{\overline{A}_{2}%
X_{2}B_{2}E_{2}}}I(\overline{A}_{2}X_{2};B_{2}|E_{2})_{\omega}.
\end{align}

Since the above inequality holds for an arbitrary probability distribution
$p_{X_{1}X_{2}}$, we conclude that%
\begin{multline}
\sup_{p_{X_{1}X_{2}}}\inf_{\rho_{\overline{A}_{1}\overline{A}_{2}X_{1}X_{2}B_{1}B_{2}E}%
}I(\overline{A}_{1}\overline{A}_{2}X_{1}X_{2};B_{1}B_{2}|E)_{\rho}\\
\leq
\sup_{p_{X_{1}}}\inf_{\omega_{\overline{A}_{1}X_{1}B_{1}E_{1}}}I(\overline{A}_{1}%
X_{1};B_{1}|E_{1})_{\omega}\\
+\sup_{p_{X_{2}}}\inf_{\omega_{\overline{A}_{2}%
X_{2}B_{2}E_{2}}}I(\overline{A}_{2}X_{2};B_{2}|E_{2})_{\omega},
\end{multline}
which is equivalent to \eqref{eq:sub-add-ineq-restr}.
\end{proof}

Monogamy of steering has been explored in \cite{Milne2014, Reid2013}. We prove
here that the restricted intrinsic steerability is monogamous in the following
sense: for a tripartite state $\rho_{ABC}$, Alice and Charlie perform
measurements on their systems and steer Bob's system. We see that their
ability to steer Bob's system is limited.

\begin{proposition}[Monogamy]
\label{sec:monogamyr}Let $\{\hat{\rho}_{B}^{a,c,x_{1},x_{2}}\}$ be an
assemblage with classical inputs $x_{1}$ and $x_{2}$\ for Alice and Charlie,
respectively, and classical outputs $a$ and $c$ for Alice and Charlie,
respectively, and obeying the following additional no-signaling constraints:%
\begin{align}
\sum_{c}\hat{\rho}_{B}^{a,c,x_{1},x_{2}}  &  =\sum_{c}\hat{\rho}%
_{B}^{a,c,x_{1},x_{2}^{\prime}}:=\hat{\theta}_{B}^{a,x_{1}}\ \ \ \ \forall
x_{2},x_{2}^{\prime},\\
\sum_{a}\hat{\rho}_{B}^{a,c,x_{1},x_{2}}  &  =\sum_{a}\hat{\rho}%
_{B}^{a,c,x_{1}^{\prime},x_{2}}:=\hat{\kappa}_{B}^{c,x_{2}}\ \ \ \ \forall
x_{1},x_{1}^{\prime},
\end{align}
such that the reduced assemblages are $\{\hat{\theta}_{B}^{a,x_{1}}%
\}_{a,x_{1}}$ and $\{\hat{\kappa}_{B}^{c,x_{2}}\}_{c,x_{2}}$.\ Then the
following monogamy inequality holds%
\begin{equation}
S^{R}(\overline{A}\overline{C};B)_{\hat{\rho}}\geq S^{R}(\overline{A};B)_{\hat{\theta}}%
+S^{R}(\overline{C};B)_{\hat{\kappa}}. \label{eq:monogamy-ineq}%
\end{equation}

\end{proposition}

\begin{proof}
This proof follows from an application of the chain rule for conditional
mutual information, much like the proof of monogamy for the squashed
entanglement \cite{KW04}. First, consider the following classical--quantum
state:%
\begin{align}
&\rho_{X_{1}X_{2}\overline{A}\overline{C}BE}:=\nonumber\\&\sum_{x_{1},x_{2},a,c}p_{X_{1}}%
(x_{1})p_{X_{2}}(x_{2})[x_{1}]\otimes\lbrack x_{2}]\otimes\lbrack
a]\otimes\lbrack c]\otimes\hat{\rho}_{BE}^{a,c,x_{1},x_{2}},
\end{align}
where $\hat{\rho}_{BE}^{a,c,x_{1},x_{2}}$ is a non-signaling extension of
$\hat{\rho}_{B}^{a,c,x_{1},x_{2}}$.\ Let%
\begin{align}
\theta_{X_{1}\overline{A}BF}  &  :=\sum_{x_{1},a}p_{X_{1}}(x_{1})[x_{1}%
]\otimes\lbrack a]\otimes\hat{\theta}_{BF}^{a,x_{1}},\\
\kappa_{X_{2}\overline{C}BG}  &  :=\sum_{x_{2},a}p_{X_{2}}(x_{2})[x_{2}%
]\otimes\lbrack c]\otimes\hat{\kappa}_{BG}^{c,x_{2}},
\end{align}
where $\hat{\theta}_{BF}^{a,x_{1}}$ is a non-signaling extension of
$\hat{\theta}_{B}^{a,x_{1}}$ and $\hat{\kappa}_{BG}^{c,x_{2}}$ is a
non-signaling extension of $\hat{\kappa}_{B}^{c,x_{2}}$. Then we have from the
chain rule for conditional mutual information that%
\begin{align}
&I(X_{1}X_{2}\overline{A}\overline{C};B|E)_{\rho}
\nonumber \\ &\quad  =I(X_{1}\overline{A};B|E)_{\rho}%
+I(X_{2}\overline{C};B|E\overline{A}X_{1})_{\rho}\\
& \quad \geq\inf_{\theta_{X_{1}\overline{A}BF}}I(X_{1}\overline{A};B|E)_{\theta}+\inf
_{\kappa_{X_{2}\overline{C}BG}}I(X_{2}\overline{C};B|G)_{\kappa}.
\end{align}
Since the above inequality holds for all non-signaling extensions $\rho
_{X_{1}X_{2}\overline{A}\overline{C}BE}$, we conclude that%
\begin{multline}
\inf_{\rho_{X_{1}X_{2}\overline{A}\overline{C}BE}}I(X_{1}X_{2}\overline{A}\overline{C};B|E)_{\rho
}\\
\geq\inf_{\theta_{X_{1}\overline{A}BF}}I(X_{1}\overline{A};B|E)_{\theta}+\inf
_{\kappa_{X_{2}\overline{C}BG}}I(X_{2}\overline{C};B|G)_{\kappa}.
\end{multline}
Optimizing the left-hand side with respect to product distributions, we find
that%
\begin{multline}
\sup_{p_{X_{1}},p_{X_{2}}}\inf_{\rho_{X_{1}X_{2}\overline{A}\overline{C}BE}}I(X_{1}%
X_{2}\overline{A}\overline{C};B|E)_{\rho}\\
\geq \inf_{\theta_{X_{1}\overline{A}BF}}I(X_{1}%
\overline{A};B|E)_{\theta}+\inf_{\kappa_{X_{2}\overline{C}BG}}I(X_{2}\overline{C}%
;B|G)_{\kappa}.
\end{multline}
The development holds for any choice of distributions $p_{X_{1}}$ and
$p_{X_{2}}$, and so we conclude that%
\begin{align}
&\sup_{p_{X_{1}},p_{X_{2}}}\inf_{\rho_{X_{1}X_{2}\overline{A}\overline{C}BE}}I(X_{1}%
X_{2}\overline{A}\overline{C};B|E)_{\rho}\nonumber \\  &\geq\sup_{p_{X_{1}}}\inf_{\theta
_{X_{1}\overline{A}BF}}I(X_{1}\overline{A};B|E)_{\theta}+\sup_{p_{X_{2}}}\inf
_{\kappa_{X_{2}\overline{C}BG}}I(X_{2}\overline{C};B|G)_{\kappa}\\
&  =S^{R}(\overline{A};B)_{\hat{\theta}}+S^{R}(\overline{C};B)_{\hat{\kappa}}.
\end{align}
Finally optimizing the left-hand side with respect to all input distributions
$p_{X_{1}X_{2}}$, we conclude \eqref{eq:monogamy-ineq}.
\end{proof}

\section{Operational interpretation\label{sec:operational}}

Let $\psi_{ABE}$ be a pure tripartite
state, $p_{X}$ a probability distribution, and $\{\Lambda_{a}^{(x)}\}_{a}$ a
positive operator-valued measure (POVM) for each $x$.\ Then $\{p_{X}%
(x)\Lambda_{a}^{(x)}\}_{a,x}$ is a POVM\ as well, representing a random choice
of the POVM\ $\{\Lambda_{a}^{(x)}\}_{a}$ according to $p_{X}$, along with
keeping a record $x$ of the choice in addition to the measurement outcome $a$.
Consider the following state resulting from performing the POVM\ on
$\psi_{ABE}$: %
\begin{multline}
\rho_{X\overline{A}BE}:=\sum_{x}|x\rangle\langle x|_{X}\otimes|a\rangle\langle
a|_{\overline{A}}\\
\otimes\operatorname{Tr}_{A}((p_{X}(x)\Lambda_{a}^{(x)}\otimes I_{BE}%
)\psi_{ABE}).
\end{multline}
Here we consider that Alice performs the measurement $\{p_{X}(x)\Lambda
_{a}^{(x)}\}_{a,x}$ on her system $A$, which results in the measurement
outcomes being placed in classical systems $X\overline{A}$. Suppose now that many
copies of the above state $\psi_{ABE}$ are available, and that Alice would
like to perform individual measurements $\{p_{X}(x)\Lambda_{a}^{(x)}\}_{a,x}$
of her systems and send all of the outcomes to Eve, who possesses the $E$
systems. Alice could certainly simply perform the measurements and send the
outcomes to Eve, but if she shares randomness with Eve, then she can simulate
the measurements in such a way as to reduce the number of classical bits she
would need to send to Eve. Furthermore, the simulation can be such that no
external party observing all of the systems could tell the difference between
the scenario in which Alice actually performs the measurements and the one in
which Alice and Eve perform a simulation of the measurements. One of the main
results of \cite{WHBH12} is that the conditional mutual information
$I(X\overline{A};B|E)_{\rho}$ is the optimal rate of classical information that
Alice needs to send to Eve in order to have a successful simulation. The
protocol that achieves this task is called measurement compression with
quantum side information \cite{WHBH12}. Thus, this information-processing task
gives an operational interpretation of the main quantity $I(X\overline
{A};B|E)_{\rho}$\ appearing in the restricted intrinsic steerability. 
We note
that our setting above, regarding the classical communication cost of
simulating steering, is rather different from the setting considered in
\cite{SABGS16}.

\section{Other possible measures}

We note here that other variations of the intrinsic steerability are possible.
Fix an assemblage $\{\hat{\rho}^{a,x}_B\}_{a,x}$. Let Eve have a non-signaling extension
of this assemblage, and we write the extended assemblage as
$\{\hat{\rho}^{a,x}_{BE}\}_{a,x}$.
Bob applies the quantum instrument consisting of trace-non-increasing completely positive maps $\{\mathcal{K}_y\}_y$, gets the outcome $y$, and publicly announces it. Then, Alice prepares the input $x$ based on $y$, and Eve performs a quantum channel $\kappa_y$ on her system. The state after this scenario is given by
\begin{multline}
\rho_{\overline{A}XB'YE}:= \sum_{x,a,y} p_{X|Y}(x|y)\op{x}{x}_X\otimes \op{a}{a}_{\ahat} \\
\otimes ({\mathcal{K}_y\otimes\kappa_y})(\hat{\rho}_{BE}^{a,x})\otimes \op{y}{y}_Y.
\end{multline}
We could then define a variation of the intrinsic steerability as 
\begin{equation}
\inf_{\rho_{\overline{A}XBYE}}\sup_{\{p(x|y),\{\mathcal{K}_y\}_y\}}I(\overline{A}X;B|EY)_\rho.
\end{equation}
This quantity however is generally larger than the intrinsic steerability, and we suspect that the definition we provided will be more useful in future applications because the definition we gave is analogous to the squashed entanglement of a channel \cite{TGW14IEEE}, which has found a number of applications in quantum information theory.  We note that it is possible to consider other restrictions that result in  a modification of the measure accordingly.
\section{Conclusion}We have introduced a quantifier for quantum steering
based on conditional quantum mutual information. It exploits the Markov-chain
structure of assemblages with a local hidden-state model, measuring the
deviation of a given assemblage from one having a local-hidden-state model.
The intrinsic steerability is a steering monotone and superadditive in
general.
This suggests that the intrinsic steerability should find applications in
protocols where steering as a resource is relevant. Also, we looked at a
restricted class of free operations. In this case, the quantity simplifies
considerably and also satisfies
additivity and monogamy. The restricted intrinsic steerability could find
applications in protocols where it suffices to consider the restricted class
of free operations.

\section{Acknowledgements}We are grateful to Rodrigo Gallego, Carl Miller,
Marco Piani, Yaoyun Shi, and Masahiro Takeoka for discussions about quantum
steering. EK acknowledges support from the Department of Physics and Astronomy
at LSU. XW and MMW acknowledge support from the NSF under Award
No.~CCF-1350397.
\bibliographystyle{unsrt}
\bibliography{steering}

\end{document}